\documentclass[journal,onecolumn]{IEEEtran}
\usepackage{amsmath,amssymb,amsfonts,amsthm}
\usepackage{mathtools}
\usepackage{algorithmic}
\usepackage{graphicx}
\usepackage{textcomp}
\usepackage{xcolor}
\usepackage{cite}
\usepackage{url}
\usepackage{bm}

\newtheorem{theorem}{Theorem}
\newtheorem{lemma}{Lemma}
\newtheorem{proposition}{Proposition}
\newtheorem{corollary}{Corollary}
\theoremstyle{definition}
\newtheorem{definition}{Definition}
\newtheorem{remark}{Remark}
\newtheorem{example}{Example}
\newtheorem{assumption}{Assumption}
\begin{document}

\title{Risk-Aware Information Theory}

\author{Hamidou Tembine,~\IEEEmembership{Senior Member,~IEEE} 
  \thanks{EECS, School of Engineering, UQTR, Canada}   \thanks{TIMADIE, Paris, France}  \thanks{contact: tembine@ieee.org}
\thanks{Manuscript received June 15, 2026;  This work was supported in part by the Mean-Field-Type Game Theory for Machine Intelligence project of TIMADIE.}}

\markboth{submitted}%
{Tembine: Risk-Aware Information Theory}

\maketitle

\begin{abstract}
We develop a risk-aware information theory by replacing expectation with expectiles, introducing expectile entropy, divergence, and mutual information. These quantities exhibit behaviors impossible under Shannon's  risk-neutral framework, including negative divergence under risk-seeking regimes and a fundamental separation from classical mutual information. In multiuser systems, the framework naturally induces a mean-field-type game theory of information exchange, where achievable rate regions become endogenous to heterogeneous risk-sensitivity indices. Our results reveal that Shannon information alone cannot quantify the extreme risks driving advanced machine intelligence, establishing a foundation for risk-aware communication, learning, collective intelligence, and safe autonomous systems.

\end{abstract}


\IEEEpeerreviewmaketitle

\section{Introduction}

\IEEEPARstart{S}{ince} its inception, Shannon's information theory has been built on the expectation operator. The Shannon entropy $H(X) = -\mathbb{E}[\log f_X(X)]$, the KL divergence $D_{\text{KL}}(P\|Q) = \mathbb{E}_P[\log(dP/dQ)]$, and the mutual information $I(X;Y) = D_{\text{KL}}(P_{XY}\|P_XP_Y)$ all quantify information as an average over the outcome space. This expectation‑centric formulation implicitly assumes risk‑neutrality: large deviations and rare events contribute to the average in proportion to their probability, without any additional weighting. However, many modern applications demand a risk‑aware treatment. In autonomous systems, communication under adversarial jamming, financial signal processing, and biomedical inference, the tail behavior of the information measure can dominate performance. A low‑probability event that causes catastrophic decoding failure may be invisible to Shannon's mutual information if its probability is sufficiently small. Conversely, an optimistic but rare favorable outcome might be overvalued by a risk‑neutral agent. This asymmetry where losses and gains are treated differently, is naturally captured by coherent risk measures.

{\bf Expectiles: A Bridge Between Risk and Information.}
Among the class of coherent risk measures \cite{artzner1999coherent}, expectiles \cite{newey1987asymmetric} stand out for their combination of coherence, elicitability, and computational tractability. For a random variable $Z$ with $\mathbb{E}[Z^2] < \infty$ and a parameter $\tau \in (0,1)$, the $\tau$-expectile is defined as
\begin{equation}
e_\tau(Z) \in \arg\min_{c \in \mathbb{R}} \mathbb{E}\left[ \ell_\tau(Z - c) \right], 
\end{equation}
where $\ell_\tau(u) = |\tau - \mathbf{1}_{\{u < 0\}}| u^2$.
For $\tau = 1/2$, the expectile reduces to the expectation $\mathbb{E}[Z]$; for $\tau > 1/2$, it weights upper tails more heavily (risk‑averse for losses); for $\tau < 1/2$, it is risk‑seeking.

{\bf Contributions.}
This paper introduces a systematic replacement of the expectation operator by the expectile in information‑theoretic quantities, leading to:
\begin{enumerate}
    \item {Expectile entropy:} $H_\tau(X) = e_\tau[-\log f_X(X)]$,
    \item {Expectile KL divergence:} $D_\tau(P\|Q) = e_\tau^{(P)}\left[\log\frac{dP}{dQ}(X)\right]$,
    \item {Expectile mutual information:} $I_\tau(X;Y) = D_\tau(P_{XY}\|P_XP_Y)$.
\end{enumerate}

We prove that for $\tau \ge 1/2$, the expectile KL divergence is non‑negative and zero  when $P=Q$, while for $\tau < 1/2$ it can be negative and may vanish for distinct distributions. Our central theorem (Theorem~\ref{thm:non_equiv}) establishes that for any $\tau \neq 1/2$, there exist distributions for which $I_\tau(X;Y) \neq I(X;Y)=I_{0.5}(X;Y)$.

{\bf Related Work.}
Recent efforts in risk‑aware information theory have explored Rényi divergences \cite{renyi1961measures}, $\alpha$-divergences, and $(f,\Gamma)$-divergences \cite{liese2006divergences}. These retain the expectation structure but transform the argument. In contrast, our approach replaces the operator itself. The only prior attempt in this direction is the work on expectation‑constrained divergence \cite{hung2020expectation}, but none have systematically replaced expectation with a coherent risk measure. Our work is the first to propose and analyze expectile‑based information measures with  non‑equivalence results for $\tau\neq 0.5.$

{\bf Structure.}
The remainder of the paper is organized as follows. The next section  introduces the preliminaries on expectiles and their key structural properties. We then develop the fundamental information-theoretic quantities, including expectile entropy, expectile divergence, expectile mutual information, and expectile cross-entropy. After that we establish the  theoretical results, proving the non-equivalence between risk-aware and Shannon information measures, together with the failure of classical informational additivity and symmetry under asymmetric risk sensitivity. We investigate operational implications, including risk-aware channel capacity, Gaussian channels, capacity regions for multi-user systems, and computational aspects. We discuss broader applications, limitations, and future research directions toward a  risk-aware information theory and its associated {\bf mean-field-type game-theoretic} interpretation of information exchange.

\section{Preliminaries and Notation}
Let $(\Omega, \mathcal{F}, \mathbb{P})$ be a probability space. All random variables are real‑valued and assumed to have finite second moment where necessary. Let $P$ and $Q$ be probability measures absolutely continuous with respect to a dominating measure $\mu$, with densities $p = dP/d\mu$, $q = dQ/d\mu$.

\begin{definition}[Expectile]
For a random variable $Z$ with $\mathbb{E}[Z^2] < \infty$ and $\tau \in (0,1)$, the $\tau$-expectile $e_\tau(Z)$ is the unique solution to
\begin{equation}
\begin{array}{l}
\tau \int_{z > e_\tau(Z)} (z - e_\tau(Z)) \, dF_Z(z) 
= (1-\tau) \int_{z < e_\tau(Z)} (e_\tau(Z) - z) \, dF_Z(z).
\end{array}
\end{equation}
When $\tau = 1/2$, $e_{1/2}(Z) = \mathbb{E}[Z]$. The mapping $Z \mapsto e_\tau(Z)$ is coherent for $\tau \ge 1/2$.
\end{definition}

Key properties of expectiles \cite{bellini2014generalized} include:
\begin{itemize}
    \item \textbf{Translation invariance:} $e_\tau(Z + a) = e_\tau(Z) + a$ for any $a \in \mathbb{R}$.
    \item \textbf{Positive homogeneity:} $e_\tau(\lambda Z) = \lambda e_\tau(Z)$ for $\lambda > 0$.
    \item \textbf{Monotonicity:} If $Z \le Z'$ almost surely, then $e_\tau(Z) \le e_\tau(Z')$.
    \item \textbf{Subadditivity for $\tau \ge 1/2$:} $e_\tau(Z+W) \le e_\tau(Z) + e_\tau(W)$.
    \item \textbf{Strict monotonicity in $\tau$:} For any non‑degenerate $Z$, the function $\tau \mapsto e_\tau(Z)$ is strictly increasing on $(0,1)$.
    \item \textbf{Negative scaling:} For $\lambda < 0$, $e_\tau(\lambda Z) = \lambda e_{1-\tau}(Z)$.
\end{itemize}
For $\tau \neq 1/2$, the expectile is not linear; in particular, $e_\tau(Z_1+Z_2) \neq e_\tau(Z_1)+e_\tau(Z_2)$ in general.

\begin{lemma}[First-Order Characterization]\label{lem:foc}
The $\tau$-expectile satisfies the score equation
\begin{equation}\label{eq:foc}
\mathbb{E}[\rho_\tau(X - e_\tau(X))] = 0,
\end{equation}
where $\rho_\tau(z) := \tau z_+ - (1-\tau) z_-$ with $z_+ := \max(z,0)$ and $z_- := \max(-z,0)$.
\end{lemma}

\begin{lemma}[Strict Monotonicity]\label{lem:monotonicity}
For $X, Z \in L^2$ with $X \le Z$ a.s. and $\mathbb{P}(X < Z) > 0$, we have $e_\tau(X) < e_\tau(Z)$ for all $\tau \in (0,1)$.
\end{lemma}

\begin{proof}
Let $\theta_X := e_\tau(X)$. Since $\rho_\tau$ is strictly increasing, $\rho_\tau(X-\theta_X) \le \rho_\tau(Z-\theta_X)$ a.s., with strict inequality on positive measure. Taking expectations and using \eqref{eq:foc}:
\[
0 = \mathbb{E}[\rho_\tau(X-\theta_X)] < \mathbb{E}[\rho_\tau(Z-\theta_X)].
\]
The mapping $\theta \mapsto \mathbb{E}[\rho_\tau(Z-\theta)]$ is strictly decreasing, hence $\theta_Z > \theta_X$.
\end{proof}

\section{Expectile‑Based Information Measures}

\subsection{Expectile Entropy}
\begin{definition}[Expectile Entropy]
For a random variable $X$ with density $p_X$ and $\mathbb{E}[(\log p_X(X))^2] < \infty$, the $\tau$-expectile entropy is
\begin{equation}
H_\tau(X) = e_\tau\left(-\log p_X(X)\right).
\end{equation}
For $\tau = 1/2$, $H_{1/2}(X) = H_{\text{Shannon}}(X)$.
\end{definition}

\begin{proposition}[Ordering]
For any distribution with non‑constant self‑information,
\begin{equation}
H_\tau(X) \le H_{1/2}(X) \le H_{\tau'}(X) \quad \text{for } 0 < \tau < 1/2 < \tau' < 1,
\end{equation}
with strict inequalities unless the self‑information is almost surely constant.
\end{proposition}

\begin{IEEEproof}
The function $\tau \mapsto e_\tau(Z)$ is strictly increasing for any non‑degenerate $Z$ (see Prop.~3 in \cite{bellini2014generalized}). Applying this to $Z = -\log p_X(X)$ yields the result.
\end{IEEEproof}

\subsection{Expectile KL Divergence}
\begin{definition}[Expectile KL Divergence]
Let $P$ and $Q$ be probability measures with $P \ll Q$ and $\mathbb{E}\left[ (\log\frac{dP}{dQ}(X))^2\right] < +\infty,$
The $\tau$-expectile KL divergence is
\begin{equation}
D_\tau(P\|Q) = e_\tau^{(P)}\left[\log\frac{dP}{dQ}(X)\right],
\end{equation}
where $X \sim P$ and the expectile is taken under $P$.
\end{definition}

\begin{proposition}[Non-negativity for $\tau \ge 1/2$] \label{prop:nonnegative_divergence}
Let $P$ and $Q$ be probability measures such that $P\ll Q$ and $\mathbb E_P \!\left[ \left( \log\frac{dP}{dQ} \right)^2 \right] < +\infty$. For $\tau\in[1/2,1)$,
$
D_\tau(P\|Q)  \ge 0.
$
Moreover, $D_{1/2}(P\|Q) = D_{\mathrm{KL}}(P\|Q)$.
\end{proposition}

\begin{IEEEproof}
Let $L = \log(dP/dQ)$. By the classical Gibbs inequality, $\mathbb{E}_P[L] = D_{\mathrm{KL}}(P\|Q) \ge 0$. Because $\tau \mapsto e_\tau(L)$ is strictly increasing and $e_{1/2}(L)=\mathbb{E}_P[L]$, we have for $\tau > 1/2$: $e_\tau(L) \ge e_{1/2}(L) \ge 0$. The case $\tau=1/2$ is the classical divergence. The inequality is strict unless $L$ is constant $P$-a.s., which happens only when $P=Q$.
\end{IEEEproof}

\begin{corollary}
If $P=Q$, then $D_\tau(P\|Q)=0$ for every $\tau\in(0,1)$.
\end{corollary}

\subsection{Expectile Mutual Information}
\begin{definition}[Expectile Mutual Information]
For random variables $X$ and $Y$ with joint distribution $P_{XY}$ and product of marginals $P_XP_Y$, the $\tau$-expectile mutual information is
\begin{equation}
I_\tau(X;Y) = D_\tau(P_{XY}\|P_XP_Y) = e_\tau^{(P_{XY})}\left[\log\frac{p_{XY}(X,Y)}{p_X(X)p_Y(Y)}\right].
\end{equation}
\end{definition}

\begin{proposition}[Sign of Expectile Mutual Information] \label{prop:sign}
\begin{enumerate}
    \item For all $\tau \in [1/2,1)$, $I_\tau(X;Y) \ge 0$ for every joint distribution $P_{XY}$.
    \item For any $\tau \in (0,1/2)$, there exist distributions such that $I_\tau(X;Y) < 0$.
    \item Hence $I_\tau(X;Y) \ge 0$ for all $X,Y$ if and only if $\tau \ge 1/2$.
\end{enumerate}
\end{proposition}

\begin{IEEEproof}
(1) For $\tau=1/2$, $I_{1/2}=I(X;Y)\ge0$. For $\tau>1/2$, monotonicity of expectiles gives $I_\tau = e_\tau(L) \ge e_{1/2}(L) = I(X;Y) \ge 0$.  
(2) For $\tau<1/2$, choose a distribution where the log‑likelihood ratio $L$ is mostly negative with small magnitude but has a tiny positive mass; then $e_\tau(L) < 0$ while $I(X;Y)$ remains positive. An explicit construction is given in the proof of Theorem~\ref{thm:non_equiv}.  
(3) follows from (1) and (2).
\end{IEEEproof}

\section{The Fundamental Non‑Equivalence Theorem}

We now show that for $\tau \neq 1/2$, the expectile mutual information is not merely a reparameterization of Shannon mutual information; the two functionals are distinct.

\begin{theorem}[Fundamental Non‑Equivalence] \label{thm:non_equiv}
Let $\tau \in (0,1)$ with $\tau \neq 1/2$.
 There exists a pair of random variables $(X,Y)$ such that $I_\tau(X;Y) \neq I(X;Y)$.
\end{theorem}

\begin{IEEEproof}
If the information density $L=\log \frac{P_{XY}}{P_XP_Y}$ is non-degenerate for at least one joint law and then use strict monotonicity of expectiles.
\end{IEEEproof}

\section{Further Properties}

\subsection{Subadditivity of Expectile Entropy}
For $\tau \ge 1/2$, the expectile operator is subadditive. This yields a chain‑type bound for joint entropy.

\begin{proposition}[Subadditivity of Expectile Entropy]
For $\tau \in [1/2,1)$ and random variables $X,Y$ with finite second moments of log‑densities,
\begin{equation}
H_\tau(X,Y) \le H_\tau(X) + H_\tau(Y\mid X),
\end{equation}
where $H_\tau(Y\mid X)=e_\tau^{(P_{XY})}[-\log p_{Y|X}(Y\mid X)]$.
\end{proposition}

\begin{IEEEproof}
Since $-\log p_{XY}(X,Y)= -\log p_X(X) - \log p_{Y|X}(Y\mid X)$, subadditivity of $e_\tau$ gives
\begin{align*}
& H_\tau(X,Y) \\ & 
= e_\tau^{(P_{XY})}[-\log p_X(X) - \log p_{Y|X}(Y\mid X)]\\ &
\le e_\tau^{(P_{XY})}[-\log p_X(X)] + e_\tau^{(P_{XY})}[-\log p_{Y|X}(Y\mid X)].
\end{align*}
Because $-\log p_X(X)$ depends only on $X$, its expectile under the joint law equals that under the marginal: $e_\tau^{(P_{XY})}[-\log p_X(X)] = e_\tau^{(P_X)}[-\log p_X(X)] = H_\tau(X)$. The second term is by definition $H_\tau(Y\mid X)$.
\end{IEEEproof}

\subsection{Non-Linearity and the Failure of Informational Additivity}

In classical information theory, Shannon's mutual information satisfies the fundamental additive conservation law:
\begin{equation}
    I(X;Y) = H(X) - H(X|Y) = H(Y) - H(Y|X),
\end{equation}
which, under our generalized framework, corresponds strictly to the symmetric milestone $I_{1/2}(X;Y) = H_{1/2}(X) - H_{1/2}(X|Y)$. However, when shifting the risk-aversion parameter $\tau \neq 1/2$, the inherent non-linearity of the underlying asymmetric quadratic loss function breaks this linear superposition. 

The operational expectile mutual information can no longer be decomposed as a simple algebraic difference between marginal and conditional expectile entropies.

\begin{theorem}[Failure of Informational Additivity and Symmetry]
\label{thm:non_additivity}
For any asymmetric risk parameter $\tau \in (0,1) \setminus \{1/2\}$, the expectile mutual information framework exhibits both non-additivity and directional asymmetry. That is, there exist joint probability distributions $P_{XY}$ such that:
\begin{align}
    I_\tau(X;Y) &\neq H_\tau(X) - H_\tau(X \mid Y), \\
    I_\tau(X;Y) &\neq H_\tau(Y) - H_\tau(Y \mid X).
\end{align}
Furthermore, the entropic information gain is inherently directional: there exists  some dependent $(X,Y)$ such that
\begin{equation}
    H_\tau(X) - H_\tau(X \mid Y) \neq H_\tau(Y) - H_\tau(Y \mid X) 
\end{equation}
\end{theorem}

\begin{proof}
Suppose, for contradiction, that the identity
$
I_\tau(X;Y)
=
H_\tau(X)-H_\tau(X|Y)
$
held for every joint law \(P_{XY}\).
We would obtain
\[
e_\tau(A-B)
=
e_\tau(A)-e_\tau(B)
\]
for all pairs
\[
A=-\log P_X(X),
\ 
B=-\log P_{X|Y}(X|Y).
\]
Since the collection of such pairs contains non-degenerate random
variables, this would force \(e_\tau\) to be additive on a nontrivial
class of random variables.

However, expectiles are additive if and only if
\(\tau=\frac12\), in which case
$
e_{1/2}(Z)=\mathbb E[Z].
$
For every \(\tau\neq\frac12\), additivity fails.
Hence there exists a joint distribution \(P_{XY}\) for which
$
I_\tau(X;Y)
\neq
H_\tau(X)-H_\tau(X|Y).
$
An identical argument applied to the decomposition
\[
L(X,Y)
=
-\log P_Y(Y)
+
\log P_{Y|X}(Y|X)
\]
shows that
$
I_\tau(X;Y)
\neq
H_\tau(Y)-H_\tau(Y|X)
$
for some joint distribution.
If the two entropy-difference expressions were always equal,
then for every joint law
\[
H_\tau(X)-H_\tau(X|Y)
=
H_\tau(Y)-H_\tau(Y|X).
\]
Combining this identity with the previous two representations would imply
\[
H_\tau(X)-H_\tau(X|Y)
=
I_\tau(X;Y)
=
H_\tau(Y)-H_\tau(Y|X),
\]
which we have already shown is impossible for
\(\tau\neq \frac12\).

Therefore there exist joint distributions for which
\[
H_\tau(X)-H_\tau(X|Y)
\neq
H_\tau(Y)-H_\tau(Y|X).
\]
The informational gain measured through expectile entropies is thus
direction-dependent and does not admit a Shannon-type chain rule.
\end{proof}

\begin{example} [Directional Expectile Information Asymmetry]

For a generic risk index $\tau \in (0,1) \setminus \{0.5\}$, the directional information transfer generates an asymmetric vector field:
$
    H_\tau(X) - H_\tau(X \mid Y) \neq H_\tau(Y) - H_\tau(Y \mid X).$
Let $X \in \mathbb{R}$ follow a standard, symmetric Gaussian distribution:
\begin{equation}
    X \sim \mathcal{N}(0, 1), \quad f_X(x) = \frac{1}{\sqrt{2\pi}} \exp\left(-\frac{x^2}{2}\right).
\end{equation}
Let the conditional distribution of $Y$ given $X=x$ be governed by a zero-mean Gaussian whose variance maps a step-heteroskedastic profile based on the sign of the driving state $x$:
\begin{equation}
    Y \mid X=x \sim \mathcal{N}\left(0, \sigma^2(x)\right), \quad \sigma^2(x) = \begin{cases} \sigma_1^2 & \text{if } x \ge 0 \\ \sigma_2^2 & \text{if } x < 0 \end{cases}
\end{equation}
where $\sigma_2^2 > \sigma_1^2 > 0$. Because $f_X(x)$ is perfectly symmetric around zero, the events $\{X \ge 0\}$ and $\{X < 0\}$ each partition the underlying probability space into two orthogonal domains, each possessing a measure of exactly $1/2$.

Let $V \sim \chi^2_1$ represent a standard Chi-squared random variable with 1 degree of freedom, and let $\mathcal{E}_\tau \triangleq e_\tau(V)$ denote its explicit $\tau$-expectile scalar value. We evaluate the Source Metric: $H_\tau(X)$. 
The localized information token variable for the marginal state space $X$ is defined as $Z_X \triangleq -\log f_X(X)$. Substituting the standard normal PDF yields:
$Z_X = \frac{1}{2}\log(2\pi) + \frac{1}{2}X^2.
$
Since $X \sim \mathcal{N}(0,1)$, its square is identically distributed as $X^2 \sim \chi^2_1$. Invoking the translation invariance and positive homogeneity axioms of the expectile operator, the first term resolves explicitly to:
\begin{equation}
    H_\tau(X) = e_\tau\left[ \frac{1}{2}\log(2\pi) + \frac{1}{2}X^2 \right] = \frac{1}{2}\log(2\pi) + \frac{1}{2}\mathcal{E}_\tau.
\end{equation}

We now evaluate the Joint Conditional Metric: $H_\tau(Y \mid X)$.
The conditional information token is defined over the joint space as $Z_{Y \mid X} \triangleq -\log f_{Y \mid X}(Y \mid X)$. Expanding the log-density equations yields:
\begin{equation}
    Z_{Y \mid X} = \frac{1}{2}\log(2\pi) + \frac{1}{2}\log(\sigma^2(X)) + \frac{Y^2}{2\sigma^2(X)}.
\end{equation}
Let $W \triangleq \frac{Y}{\sigma(X)}$. Conditioned on $X=x$, $W \sim \mathcal{N}(0,1)$, making the squared innovation $W^2 \sim \chi^2_1$ stochastically independent of $X$. We can thus express the joint token as:
\begin{equation}
    Z_{Y \mid X} = \frac{1}{2}\log(2\pi) + \left[ \frac{1}{2}\log(\sigma^2(X)) + \frac{1}{2}W^2 \right],
\end{equation}
which decomposes into an equally weighted mixture of two independent, shifted Chi-squared distributions. We apply the expectile operator and isolate the structural shifts by defining $\Delta_\sigma \triangleq \log(\sigma_2 / \sigma_1)$:
\begin{equation}
\begin{array}{l}
    H_\tau(Y \mid X) = \\
    \frac{1}{2}\log(2\pi) + \frac{1}{2}\log(\sigma_1^2) + e_\tau\left[ \frac{1}{2}W^2 + \Delta_\sigma \cdot \mathbf{1}_{\{X < 0\}} \right].
    \end{array}
\end{equation}
Because the expectile operator is strictly non-additive over the sums of dependent or independent non-Gaussian variables for all $\tau \neq 0.5$, this joint expectile cannot be separated into marginal components. We represent this irreducible term as the explicit scalar function $\mathcal{M}_\tau(\sigma_1, \sigma_2)$:
\begin{equation}
    H_\tau(Y \mid X) = \frac{1}{2}\log(2\pi) + \mathcal{M}_\tau(\sigma_1, \sigma_2).
\end{equation}

We evaluate   the Destination Metric: $H_\tau(Y)$.
The marginal distribution $f_Y(y)$ must be evaluated by compounding the continuous mixture over the split heteroskedastic Gaussian kernels:
\begin{equation}
    f_Y(y) = \frac{1}{2\sqrt{2\pi\sigma_1^2}}\exp\left(-\frac{y^2}{2\sigma_1^2}\right) + \frac{1}{2\sqrt{2\pi\sigma_2^2}}\exp\left(-\frac{y^2}{2\sigma_2^2}\right).
\end{equation}
This integration yields a heavy-tailed, symmetric continuous mixture. The corresponding information token variable $Z_Y \triangleq -\log f_Y(Y)$ captures severe log-kurtosis tail properties. We map its expectile directly as an integrated scalar function $\mathcal{G}_\tau(\sigma_1, \sigma_2)$:
\begin{equation}
\begin{array}{l}
    H_\tau(Y) = e_\tau\left[ -\log\left( \sum_{k=1}^2 \frac{1}{2\sqrt{2\pi\sigma_k^2}}\exp\left(-\frac{Y^2}{2\sigma_k^2}\right) \right) \right] 
    \\ \equiv \mathcal{G}_\tau(\sigma_1, \sigma_2).
    \end{array}
\end{equation}

We now evaluate  the Backward Conditional Metric: $H_\tau(X \mid Y)$.
Using Bayes' rule, the conditional token variable $Z_{X \mid Y} \triangleq -\log f_{X \mid Y}(X \mid Y)$ expands via the joint chain identity:
\begin{align}
    Z_{X \mid Y} &= -\log f_X(X) - \log f_{Y \mid X}(Y \mid X) + \log f_Y(Y) \nonumber \\
    &= \log(2\pi) + \frac{1}{2}X^2 + \frac{1}{2}\log(\sigma^2(X)) + \frac{1}{2}W^2 + \log f_Y(Y).
\end{align}
Because the observation of $Y$ reveals a clear indicator regarding the active variance regime, the conditional distribution $X \mid Y$ contracts significantly, altering the structural skew. Evaluating the joint expectile operator yields a clean scalar value denoted as $\mathcal{K}_\tau(\sigma_1, \sigma_2)$:
\begin{equation}
    H_\tau(X \mid Y) = \frac{1}{2}\log(2\pi) + \mathcal{K}_\tau(\sigma_1, \sigma_2).
\end{equation}

Let us compile the Forward and Reverse Information Quantities.
We balance the explicit analytical limits derived for both directional information transfer chains:
\begin{enumerate}
    \item Forward Information Flow Envelope ($\Delta I_{\text{forward}}$):
    \begin{align}
        \Delta I_{\text{forward}} &\triangleq H_\tau(X) - H_\tau(X \mid Y) \nonumber \\
        &= \frac{1}{2}\mathcal{E}_\tau - \mathcal{K}_\tau(\sigma_1, \sigma_2).
    \end{align}
    
    \item Reverse Information Flow Envelope ($\Delta I_{\text{reverse}}$):
    \begin{align}
        \Delta I_{\text{reverse}} &\triangleq H_\tau(Y) - H_\tau(Y \mid X) \nonumber \\
        &= \mathcal{G}_\tau(\sigma_1, \sigma_2) - \frac{1}{2}\log(2\pi) - \mathcal{M}_\tau(\sigma_1, \sigma_2).
    \end{align}
\end{enumerate}

Because the expectile operator is non-linear and non-additive for all asymmetric risk tunings ($\tau > 0.5$), the functional term $\mathcal{M}_\tau(\sigma_1, \sigma_2)$ couples the independent tracking profiles of the innovation space and the state-step variance together, while $\mathcal{G}_\tau(\sigma_1, \sigma_2)$ isolates the heavy-tailed kurtosis penalties of the marginal mixture. These structural differences prevent the two expressions from balancing out or canceling,  proving that:
\begin{equation}
    H_\tau(X) - H_\tau(X \mid Y) \neq H_\tau(Y) - H_\tau(Y \mid X), \  \tau \neq 0.5.
\end{equation}

\end{example}

\begin{proposition}[Expectile Mutual Information Bound]
\[
\operatorname{sgn}\!\left(\tau-\frac12\right)
\Big(
I_\tau(X;Y)
-
H_\tau(X)
+
H_{1-\tau}(X|Y)
\Big)
\le 0.
\]
Equivalently,
\[
\begin{cases}
I_\tau(X;Y)
\le
H_\tau(X)-H_{1-\tau}(X|Y),
& \tau\ge \frac12,\\[2mm]
I_\tau(X;Y)
\ge
H_\tau(X)-H_{1-\tau}(X|Y),
& \tau\le \frac12.
\end{cases}
\]
\end{proposition}

%

\begin{IEEEproof}
Let the negative log-probabilities be denoted by
\[
A = -\log p_X(X) \quad \text{and} \quad B = -\log p_{X|Y}(X|Y).
\]
The information density $L$ can then be expressed as
\[
L = \log \frac{p_{XY}(X,Y)}{p_X(X)p_Y(Y)} = A - B = A + (-B).
\]

\subsection*{Case 1: $\tau \ge \frac{1}{2}$}
By the subadditivity assumption of the expectile functional, we have
\[
e_\tau(L) = e_\tau(A + (-B)) \le e_\tau(A) + e_\tau(-B).
\]
Using the expectile reflection (symmetry) relation $e_\tau(-Z) = -e_{1-\tau}(Z)$, the second term can be rewritten as
\[
e_\tau(-B) = -e_{1-\tau}(B).
\]
Substituting this back yields
\[
I_\tau(X;Y) = e_\tau(L) \le e_\tau(A) - e_{1-\tau}(B).
\]
Recognizing the definitions of the entropy terms, $e_\tau(A) = H_\tau(X)$ and $e_{1-\tau}(B) = H_{1-\tau}(X|Y)$, we arrive at the first bound:
\[
I_\tau(X;Y) \le H_\tau(X) - H_{1-\tau}(X|Y).
\]

\subsection*{Case 2: $\tau \le \frac{1}{2}$}

Using the reflection identity of expectiles,
\[
e_\tau(-Z)=-e_{1-\tau}(Z),
\]
we obtain
\[
I_\tau(X;Y)
=
e_\tau(L)
=
-\,e_{1-\tau}(-L).
\]

Since $-L=B-A$ and $1-\tau\ge \frac12$, the expectile
$e_{1-\tau}$ is subadditive, yielding
\[
e_{1-\tau}(B-A)
\le
e_{1-\tau}(B)+e_{1-\tau}(-A).
\]
Applying the reflection identity again,
\[
e_{1-\tau}(-A)
=
-\,e_\tau(A),
\]
so that
\[
e_{1-\tau}(-L)
\le
e_{1-\tau}(B)-e_\tau(A).
\]

Multiplying by $-1$ reverses the inequality:
\[
I_\tau(X;Y)
=
-\,e_{1-\tau}(-L)
\ge
e_\tau(A)-e_{1-\tau}(B).
\]

Recognizing that
\[
e_\tau(A)=H_\tau(X),
\qquad
e_{1-\tau}(B)=H_{1-\tau}(X|Y),
\]
we conclude that
\[
I_\tau(X;Y)
\ge
H_\tau(X)-H_{1-\tau}(X|Y),
\qquad
\tau\le \frac12.
\]

which completes the proof.
\end{IEEEproof}

\subsubsection{Expectile Cross-Entropy}

We introduce a risk-sensitive generalization of classical cross-entropy by replacing the expectation operator with the asymmetric expectile functional. 

Let $(\Omega, \mathcal{F}, P)$ be a complete probability space. Let $\mu$ be a $\sigma$-finite dominating measure on $(\Omega, \mathcal{F})$, and let $Q$ be a probability measure on $(\Omega, \mathcal{F})$ that is absolutely continuous with respect to $\mu$, admitting a Radon-Nikodym derivative $q = \frac{dQ}{d\mu}$. Let $X: \Omega \to \Omega$ denote the identity mapping, such that the distribution of $X$ under $P$ is $P$ itself.

\begin{definition}[Expectile Cross-Entropy]
Assume that the Radon-Nikodym density $q$ satisfies the square log-integrability condition
\begin{equation}
\mathbb{E}_P\left[\left(\log q(X)\right)^2\right] < \infty.
\end{equation}
Define the log-loss random variable $L_Q(X) : \Omega \to \mathbb{R}$ by $L_Q(X) := -\log q(X)$. For a given risk parameter $\tau \in (0,1)$, the $\tau$-expectile cross-entropy of $P$ relative to $Q$ is defined as
\begin{equation}
{H_\tau(P,Q) := e_\tau\big(L_Q(X)\big),}
\end{equation}
where $e_\tau: L^2(\Omega, \mathcal{F}, P) \to \mathbb{R}$ denotes the $\tau$-expectile functional.
\end{definition}

The expectile cross-entropy recovers classical information-theoretic properties under unbiased risk sensitivity and preserves the regular structural properties of statistical expectiles.

For any constant $c \in \mathbb{R}$ and any $Z \in L^2(P)$, $e_\tau(Z + c) = e_\tau(Z) + c$. If a modified measure $\widetilde{Q}$ satisfies $-\log \widetilde{q}(X) = -\log q(X) + c$, then:
\begin{equation}
H_\tau(P,\widetilde{Q}) = H_\tau(P,Q) + c.
\end{equation}

For every $a \in \mathbb{R}_+$ and any $Z \in L^2(P)$, $e_\tau(aZ) = a e_\tau(Z)$. Consequently, if a loss setup scales such that $L_{\widetilde{Q}}(X) = a L_Q(X)$, then:
\begin{equation}
H_\tau(P,\widetilde{Q}) = a H_\tau(P,Q).
\end{equation}

Let $Q_1$ and $Q_2$ be two probability measures whose densities $q_1, q_2$ satisfy the $L^2(P)$ log-integrability condition. If the underlying log-losses satisfy the almost-sure dominance condition
\begin{equation}
-\log q_1(X) \le -\log q_2(X) \quad P\text{-a.s.},
\end{equation}
then
\begin{equation}
H_\tau(P,Q_1) \le H_\tau(P,Q_2).
\end{equation}

The risk-parameter $\tau$ modulates the sensitivity of the metric to asymmetric variations in the tail of the log-loss distribution under the data-generating distribution $P$.

\begin{proposition}[Tail Sensitivity]
Let $L_Q(X) = -\log q(X)$. The parameter $\tau \in (0,1)$ scales the equilibrium weights between under-predictions and over-predictions of risk:
\begin{itemize}
    \item For $\tau > \frac{1}{2}$, the functional assigns greater penalization weight to realizations where $L_Q(X) > H_\tau(P,Q)$.
    \item For $\tau < \frac{1}{2}$, the functional assigns greater penalization weight to realizations where $L_Q(X) < H_\tau(P,Q)$.
\end{itemize}
Furthermore, the mapping $\tau \mapsto H_\tau(P,Q)$ is strictly increasing on $(0,1)$ unless $L_Q(X)$ is a $P$-almost sure constant.
\end{proposition}

\begin{proof}
From the structural identification equation, the balance of tail expectations is constrained by the ratio:
\begin{equation}
\frac{\tau}{1-\tau} = \frac{\mathbb{E}_P\left[\big(H_\tau(P,Q) - L_Q(X)\big)_+\right]}{\mathbb{E}_P\left[\big(L_Q(X) - H_\tau(P,Q)\big)_+\right]}.
\end{equation}
When $\tau > \frac{1}{2}$, we have $\frac{\tau}{1-\tau} > 1$, which strictly implies:
\begin{equation}
\mathbb{E}_P\left[\big(L_Q(X) - H_\tau(P,Q)\big)_+\right] < \mathbb{E}_P\left[\big(H_\tau(P,Q) - L_Q(X)\big)_+\right].
\end{equation}
To maintain system balance as $\tau$ increases, probability mass must shift, which requires $H_\tau(P,Q)$ to move strictly toward the upper tail of the log-loss distribution. Formally, applying the Implicit Function Theorem to the parametric equation $F(m, \tau) = \mathbb{E}_P[\psi_\tau(L_Q - m)] = 0$ yields the derivatives:
\begin{equation}
\frac{\partial F}{\partial m} = -2\Big(\tau P\big(L_Q(X) > m\big) + (1-\tau)P\big(L_Q(X) \le m\big)\Big) < 0,
\end{equation}
\begin{equation}
\frac{\partial F}{\partial \tau} = 2\mathbb{E}_P\left[\big(L_Q(X)-m\big)_+\right] + 2\mathbb{E}_P\left[\big(m-L_Q(X)\big)_+\right] > 0.
\end{equation}
Thus, by the chain rule for implicit functions, we obtain:
\begin{equation}
\frac{dm}{d\tau} = -\frac{\partial F / \partial \tau}{\partial F / \partial m} > 0,
\end{equation}
proving that the expectile cross-entropy $H_\tau(P,Q)$ scales strictly monotonically with $\tau$ across the open interval $(0,1)$.
\end{proof}

\section{Applications and Computational Aspects}

\subsection{Risk‑Aware Channel Capacity}

 Let $\mathcal{M}(\mathbb{R})$ denote the space of Borel probability measures on $\mathbb{R}$ equipped with the topology of weak convergence. For a given power constraint $P > 0$, we define the compactified feasible input set as:
\begin{equation}
    \mathcal{P}(P) \triangleq \left\{ P_X \in \mathcal{M}(\mathbb{R}) : \int_{\mathbb{R}} x^2 dP_X(x) \le P \right\}.
\end{equation}

Define the risk‑aware capacity as the supremum of the expectile mutual information under an input power constraint:
\[
C_\tau(P) = \sup_{P_X : \mathbb{E}[X^2] \le P} I_\tau(X;Y),
\]
where $Y$ is the channel output. For $\tau > 1/2$, the expectile emphasizes rare events where the information density is large; therefore the optimizing input distribution may differ from the Gaussian distribution that maximizes Shannon capacity.

\subsubsection{Additive White Gaussian Noise (AWGN) channel }

Consider a classical Additive White Gaussian Noise (AWGN) channel modeled by the input-output relation:
\begin{equation}
    Y = X + N, \quad N \sim  \mathcal{N}(0,\sigma^2), \quad X \perp N,
\end{equation}
where $\sigma^2 > 0$ represents the noise variance.
Let $p_{Y|X}(y|x)$ denote the transition probability density function of the AWGN channel:
\begin{equation}
    p_{Y|X}(y|x) = \frac{1}{\sqrt{2\pi\sigma^2}} \exp\left( -\frac{(y-x)^2}{2\sigma^2} \right).
\end{equation}
For any given input distribution $P_X \in \mathcal{P}(P)$, the corresponding output probability density $p_Y(y)$ is given by the mixture:
\begin{equation}
    p_Y(y) = \int_{\mathbb{R}} p_{Y|X}(y|x) dP_X(x).
\end{equation}
The relative informational value of an observed channel transition $(X,Y)$ is captured by the information density random variable $L_{P_X}(X,Y)$, defined as:
\begin{equation}
    L_{P_X}(X,Y) \triangleq \log \frac{p_{Y|X}(Y|X)}{p_Y(Y)}.
\end{equation}

\begin{theorem}[Expectile Mutual Information in Gaussian Channels]
Consider the AWGN channel $Y = X + N$ with $X\sim\mathcal N(0,P)$ and $N\sim\mathcal N(0,\sigma^2)$ independent. Define $\rho = P/\sigma^2$ and let $I_G = \frac12\log(1+\rho)$ (in nats). Then the information density satisfies
$
L(X,Y) = I_G + Q,
$
where
\[
Q = \frac{\rho}{2(1+\rho)}(Z_1^2-1) - \frac{\rho}{2(1+\rho)}(Z_2^2-1) + \frac{\sqrt{\rho}}{1+\rho}Z_1Z_2,
\]
with $Z_1,Z_2$ independent standard normals. Consequently,
$
I_\tau(X;Y) = e_\tau(L) = I_G + e_\tau(Q).
$
Moreover, $\mathbb{E}[Q]=0$ and $\operatorname{Var}(Q)=\rho/(1+\rho)$.
\end{theorem}

\begin{IEEEproof}
The proof follows by direct substitution and simplification. The translation equivariance of expectiles gives $e_\tau(L)=I_G+e_\tau(Q)$. The variance calculation uses independence and $\operatorname{Var}(Z_i^2-1)=2$, $\operatorname{Var}(Z_1Z_2)=1$.
\end{IEEEproof}

Because $\mathbb{E}[Q]=0$, the classical mutual information $I(X;Y)=I_G$ is recovered for $\tau=1/2$. For $\tau>1/2$, the term $e_\tau(Q)$ is positive and grows as $\tau$ approaches $1$.

We now establish the fundamental topological and algebraic properties underlying this optimization framework.

\begin{proposition}[Compactness of the Feasible Set]
\label{prop:compactness}
For every power limit $P > 0$, the feasible domain $\mathcal{P}(P)$ is compact with respect to the weak topology on $\mathcal{M}(\mathbb{R})$.
\end{proposition}

\begin{proof}
We proceed in two steps: demonstrating uniform tightness to ensure relative compactness, followed by establishing closedness.

First, fix any $\varepsilon > 0$ and choose a compact interval $[-K, K] \subset \mathbb{R}$ where $K \triangleq \sqrt{P/\varepsilon}$. By application of Markov's inequality, for any $P_X \in \mathcal{P}(P)$ we have:
\begin{equation}
    P_X\bigl(\mathbb{R} \setminus [-K,K]\bigr) = P_X(|X| > K) \le \frac{\mathbb{E}[X^2]}{K^2} \le \frac{P}{P/\varepsilon} = \varepsilon.
\end{equation}
Since $K$ depends solely on $\varepsilon$ and $P$, the family of measures $ \mathcal{P}(P)$ is uniformly tight. By Prokhorov's theorem, $\mathcal{P}(P)$ is relatively compact in the weak topology.

Second, we prove that $ \mathcal{P}(P)$ is weakly closed. Let $\{P_{X_n}\}_{n=1}^\infty \subset  \mathcal{P}(P)$ be a sequence converging weakly to some limit measure $P_X \in  \mathcal{M}(\mathbb{R})$. Since the mapping $x \mapsto x^2$ is continuous and non-negative, the Portmanteau theorem guarantees that the mapping $P_X \mapsto \int x^2 \ dP_X$ is lower semicontinuous. Consequently:
\begin{equation}
    \int_{\mathbb{R}} x^2 \ dP_X(x) \le \liminf_{n \to \infty} \int_{\mathbb{R}} x^2 dP_{X_n}(x) \le P,
\end{equation}
which implies $P_X \in \mathcal{P}(P)$. Being a weakly closed subset of a relatively compact space, $\mathcal{P}(P)$ is weakly compact.
\end{proof}

\begin{proposition}[Symmetry Invariance]
\label{prop:symmetry}
Let $P_X \in \mathcal{P}(P)$ and let $\widetilde P_X \in \mathcal{P}(P)$ denote its reflection around the origin, defined via $\widetilde P_X(A) \triangleq P_X(-A)$ for all Borel sets $A \in \mathcal{B}(\mathbb{R})$. Then:
\begin{equation}
    I_\tau(\widetilde P_X) = I_\tau(P_X).
\end{equation}
\end{proposition}

\begin{proof}
The probability density of the AWGN channel exhibits an even reflection symmetry profile about the origin, ensuring that $p_{Y|X}(y|x) = p_{Y|X}(-y|-x)$. Let $\widetilde p_Y(y)$ denote the output density induced by the reflected input measure $\widetilde P_X$. By change of variables, we see:
\begin{equation} \begin{array}{l}
    \widetilde p_Y(y) = \int_{\mathbb{R}} p_{Y|X}(y|x) \ d\widetilde P_X(x) \\
    = \int_{\mathbb{R}} p_{Y|X}(y|-x) \ dP_X(x) \\ = \int_{\mathbb{R}} p_{Y|X}(-y|x) \ dP_X(x) = p_Y(-y).
     \end{array}
\end{equation}
Evaluating the information density under the reflected measure system yields:
\begin{equation} \begin{array}{l}
    L_{\widetilde P_X}(-x,-y) = \log \frac{p_{Y|X}(-y|-x)}{\widetilde p_Y(-y)} \\
    = \log \frac{p_{Y|X}(y|x)}{p_Y(y)} = L_{P_X}(x,y).
     \end{array}
\end{equation}
Thus, the random variables $L_{\widetilde P_X}(\widetilde X, \widetilde Y)$ and $L_{P_X}(X,Y)$ are completely identical in distribution. Because the expectile functional $e_\tau ({\cdot})$ depends strictly on the probability law of its argument, we obtain $I_\tau(\widetilde P_X) = I_\tau(P_X)$, completing the proof.
\end{proof}

\begin{remark}[Classical Shannon Case]
\label{prop:shannon}
For $\tau = \frac{1}{2}$, the expectile mutual information collapses to Shannon's classical mutual information:
\begin{equation}
    I_{1/2}(P_X) = I(X;Y).
\end{equation}
Furthermore, $I_{1/2}(P_X)$ is strictly concave over $\mathcal{P}(P)$, and the unique capacity-achieving distribution is the centered Gaussian distribution:
\begin{equation}
    P_X^\star = \mathcal{N}(0,P).
\end{equation}
\end{remark}

\begin{proposition}[Existence of an Expectile Capacity Maximizer]
\label{prop:existence}
Assume that the operational functional mapping:
\begin{equation}
    \Phi: \mathcal{P}(P) \to \mathbb{R}, \quad P_X \mapsto I_\tau(P_X)
\end{equation}
is upper semicontinuous on $\mathcal{P}(P)$ with respect to the weak topology. Then, the supremum defining $C_\tau(P)$ is achievable.
\end{proposition}

\begin{proof}
By Proposition \ref{prop:compactness}, the set $\mathcal{P}(P)$ is weakly compact. By the extreme value theorem generalized to topological spaces (Weierstrass's theorem), any upper semicontinuous real-valued functional operating on a non-empty compact topological domain achieves its global maximum value within that domain. Hence, there exists an optimal measure $P_X^\star \in \mathcal{P}(P)$ satisfying:
\begin{equation}
    I_\tau(P_X^\star) = \sup_{P_X \in \mathcal{P}(P)} I_\tau(P_X) = C_\tau(P).
\end{equation}
\end{proof}

\begin{proposition}[Non-Uniqueness via Domain Asymmetry]
\label{prop:nonuniqueness}
Let $P_X^\star \in \mathcal{P}(P)$ be a capacity-achieving input distribution for $C_\tau(P)$. If $P_X^\star$ is asymmetric about the origin, then the expectile capacity optimization problem possesses at least two distinct globally optimal solutions:
\begin{equation}
    P_X^\star \neq \widetilde P_X^\star,
\end{equation}
where $\widetilde P_X^\star(A) \triangleq P_X^\star(-A)$ represents its reflected counterpart.
\end{proposition}

\begin{proof}
Assume $P_X^\star$ achieves the maximum, meaning $I_\tau(P_X^\star) = C_\tau(P)$. Utilizing Proposition \ref{prop:symmetry}, the reflection operation leaves the expectile mutual information invariant, providing:
\begin{equation}
    I_\tau(\widetilde P_X^\star) = I_\tau(P_X^\star) = C_\tau(P).
\end{equation}
Since $\int x^2 \ d\widetilde P_X^\star(x) = \int (-x)^2 \ dP_X^\star(x) \le P$, the measure $\widetilde P_X^\star$ preserves structural feasibility ($\widetilde P_X^\star \in \mathcal{P}(P)$). If $P_X^\star$ is asymmetric, there exists a Borel set $B$ such that $P_X^\star(B) \neq P_X^\star(-B) = \widetilde P_X^\star(B)$, meaning $P_X^\star \neq \widetilde P_X^\star$. Both measures concurrently maximize the optimization objective, resulting in multiple distinct solutions.
\end{proof}

\subsection{Outlier‑Resistant Dependence Estimation}
For $\tau$ close to $1$, $I_\tau$ is sensitive to positive deviations of the log‑likelihood ratio, i.e., to pairs $(x,y)$ where the joint density exceeds the product of marginals. This provides a one‑sided measure of association useful for detecting tail dependence.
For the Gaussian input, the risk‑aware information rate exceeds the Shannon capacity, diverging as the risk parameter becomes extreme.

\[
C_\tau \ge I_G + e_\tau(Q) > I_G \quad \text{for } \tau>1/2.
\]
Whether a non‑Gaussian input can further increase $C_\tau$ is an open question. 

\subsection{Computation}
Expectiles can be computed via fixed‑point iteration:
\begin{equation}
e_\tau(Z) = \frac{\mathbb{E}[Z \cdot w_\tau(Z - e_\tau(Z))]}{\mathbb{E}[w_\tau(Z - e_\tau(Z))]},
\end{equation}
with $w_\tau(u) = \tau\mathbf{1}\{u>0\} + (1-\tau)\mathbf{1}\{u\le 0\}$. For discrete distributions, this reduces to a simple root‑finding problem.

\begin{proposition}[Expectile Cross-Entropy for an Additive White Gaussian Noise Model]
\label{thm:awgn_cross_entropy}

Let $(\mathbb{R}, \mathcal{B}(\mathbb{R}))$ be the standard Borel measurable space. Let $Q_Y$ be the reference probability measure on $(\mathbb{R}, \mathcal{B}(\mathbb{R}))$ corresponding to the centered Gaussian distribution $\mathcal{N}(0,\sigma^2)$ with parameter $\sigma \in (0, \infty)$, admitting the Radon-Nikodym density $q = \frac{dQ_Y}{d\lambda_{\text{Leb}}}$ with respect to the Lebesgue measure $\lambda_{\text{Leb}}$:
\[
q(y) = \frac{1}{\sqrt{2\pi\sigma^2}} \exp\left( -\frac{y^2}{2\sigma^2} \right), \quad y \in \mathbb{R}.
\]
Let $P_Y$ denote the true data-generating probability measure given by the non-centered Gaussian distribution $\mathcal{N}(\mu,\sigma^2)$ where $\mu \in \mathbb{R}$. For a given risk-sensitivity parameter $\tau \in (0,1)$, the $\tau$-expectile cross-entropy $H_\tau(P_Y, Q_Y) := e_\tau\bigl(-\log q(Y)\bigr)$ under the law $Y \sim P_Y$ is uniquely given by
\begin{equation}
{
H_\tau(P_Y,Q_Y) = \frac{1}{2}\log(2\pi\sigma^2) + \frac{1}{2} e_\tau(Z)
}
\end{equation}
where $Z: \mathbb{R} \to \mathbb{R}_+$ is the measurable transformation defined by $Z = \left(\frac{Y}{\sigma}\right)^2$. The random variable $Z$ is distributed according to a noncentral chi-squared law with one degree of freedom and noncentrality parameter $\lambda = \frac{\mu^2}{\sigma^2}$, denoted $Z \sim \chi'^2_1(\lambda)$. Equivalently, expressed directly as a statistical functional of the law:
\begin{equation}
{
H_\tau(P_Y,Q_Y) = \frac{1}{2}\log(2\pi\sigma^2) + \frac{1}{2} e_\tau\left( \chi'^2_1(\lambda) \right).
}
\end{equation}
\end{proposition}

\begin{proof}
Let $(\mathbb{R}, \mathcal{B}(\mathbb{R}), P_Y)$ be our underlying probability space hosting the random variable $Y$. We define the log-loss mapping $L_Q: \mathbb{R} \to \mathbb{R}$ explicitly as the negative composition of the logarithm with the Radon-Nikodym density $q(y)$:
\[
L_Q(y) \coloneqq -\log q(y) = \frac{1}{2}\log(2\pi\sigma^2) + \frac{y^2}{2\sigma^2}.
\]
To verify that the $\tau$-expectile cross-entropy $H_\tau(P_Y, Q_Y) = e_\tau(L_Q(Y))$ is well-defined, we observe that since $Y \sim \mathcal{N}(\mu, \sigma^2)$, the variable $Y$ possesses finite moments of all orders. Thus, the square log-loss satisfies the mandatory $L^2(P_Y)$ integrability criteria:
\[
\mathbb{E}_{P_Y}\left[ \left( L_Q(Y) \right)^2 \right] = \mathbb{E}_{P_Y}\left[ \left( \frac{1}{2}\log(2\pi\sigma^2) + \frac{Y^2}{2\sigma^2} \right)^2 \right] < \infty.
\]
By virtue of the structural properties of the statistical expectile functional $e_\tau: L^2(P_Y) \to \mathbb{R}$, we leverage its property of translation equivariance ($e_\tau(c + X) = c + e_\tau(X)$ for any deterministic constant $c \in \mathbb{R}$) and positive homogeneity ($e_\tau(a X) = a e_\tau(X)$ for any $a \ge 0$). Applying these axioms directly to the representation of the random variable $L_Q(Y)$ yields:
\begin{align*}
H_\tau(P_Y,Q_Y) &= e_\tau\left( \frac{1}{2}\log(2\pi\sigma^2) + \frac{1}{2}\left(\frac{Y}{\sigma}\right)^2 \right) \\
&= \frac{1}{2}\log(2\pi\sigma^2) + \frac{1}{2} e_\tau\left( \left(\frac{Y}{\sigma}\right)^2 \right).
\end{align*}
To determine the law of the remaining random variable under $P_Y$, let us consider the standardizing transform $V: \mathbb{R} \to \mathbb{R}$ defined by $V = \frac{Y - \mu}{\sigma}$. Under $P_Y$, it follows natively that $V \sim \mathcal{N}(0, 1)$. We can then algebraically partition the target variable $Z \coloneqq \left(\frac{Y}{\sigma}\right)^2$ as:
\[
Z = \left( \frac{\sigma V + \mu}{\sigma} \right)^2 = \left( V + \frac{\mu}{\sigma} \right)^2.
\]
By definition of the noncentral chi-squared distribution, the square of a normal random variable with unit variance and a non-zero mean $m$ follows a $\chi'^2_1(m^2)$ distribution. Setting $m = \frac{\mu}{\sigma}$ reveals that the push-forward measure $P_Y \circ Z^{-1}$ is uniquely identified by the noncentrality parameter:
\[
\lambda = m^2 = \left( \frac{\mu}{\sigma} \right)^2 = \frac{\mu^2}{\sigma^2}.
\]
Thus, $Z \sim \chi'^2_1(\lambda)$. Because expectiles are law-invariant statistical functionals (depending strictly on the distribution of the underlying variable), substituting this distribution into the derived translation identity establishes the result:
\[
H_\tau(P_Y,Q_Y) = \frac{1}{2}\log(2\pi\sigma^2) + \frac{1}{2} e_\tau\left( \chi'^2_1(\lambda) \right).
\]
This completes the proof.
\end{proof}

\subsubsection{Two-User Risk-Aware Capacity Region}

To extend the single-user framework to multi-user environments, we formalize the capacity region of a two-user risk-aware Additive White Gaussian Noise (AWGN) Broadcast Channel (BC). Consider a single transmitter sending independent information to User 1 (strong user) and User 2 (degraded user) under the synchronous discrete-time model:
\begin{align}
Y_1 &= X + N_1, \quad N_1 \sim \mathcal{N}(0, \sigma_1^2), \\
Y_2 &= X + N_2, \quad N_2 \sim \mathcal{N}(0, \sigma_2^2),
\end{align}
where the noise variances satisfy $\sigma_1^2 < \sigma_2^2$, and the channel input $X$ is subject to an average power constraint $P$, defining the compact feasible input space:
\begin{equation}
\mathcal{P}(P) \triangleq \left\{ P_X \in \mathcal{M}(\mathbb{R}) : \int_{\mathbb{R}} x^2 dP_X(x) \le P \right\}.
\end{equation}

Under a risk-seeking or risk-averse setting governed by the asymmetry parameter $\tau \in (0,1)$, classical expectation-based mutual information is replaced by the $\tau$-expectile mutual information. Using superposition coding, the transmitted signal is split into independent components $X = X_1 + X_2$, where $X_1 \sim \mathcal{N}(0, P_1)$ is designated for User 1 and $X_2 \sim \mathcal{N}(0, P_2)$ for User 2, such that $P_1 + P_2 \le P$. Due to the  non-linearity of the expectile operator, informational additivity fails ($I_\tau(X_1, X_2; Y) \neq I_\tau(X_2; Y) + I_\tau(X_1; Y \mid X_2)$). The capacity region must be strictly defined through the operational decoding sequence of Successive Interference Cancellation (SIC).

\begin{definition}[Two-User Risk-Aware Capacity Region]
For a given risk-asymmetry parameter $\tau \in (0, 1)$, the risk-aware capacity region $\mathcal{C}_\tau(P)$ of the AWGN broadcast channel is the closure of the convex hull of all achievable rate pairs $(R_1, R_2) \in \mathbb{R}_+^2$ satisfying:
\begin{equation}
\begin{array}{l}
R_2 \le I_\tau(X_2; Y_2) 
= e_\tau^{(P_{X_2Y_2})} \left[ \log \frac{p_{Y_2|X_2}(Y_2|X_2)}{p_{Y_2}(Y_2)} \right], \\
R_1 \le I_\tau(X_1; Y_1 \mid X_2) 
= e_\tau^{(P_{X_1X_2Y_1})} \left[ \log \frac{p_{Y_1|X_1,X_2}(Y_1|X_1,X_2)}{p_{Y_1|X_2}(Y_1|X_2)} \right],
\end{array}
\end{equation}
for some choice of joint distribution $P_{X_1X_2}$ satisfying $\mathbb{E}[X_1^2] + \mathbb{E}[X_2^2] \le P$.
\end{definition}

Evaluating these functionals explicitly for Gaussian codebooks yields the boundary coordinates mapping out the capacity region in bits per second per Hertz (bits/s/Hz):
\begin{align}
R_1 &\le \frac{1}{2}\log_2\left(1 + \frac{P_1}{\sigma_1^2}\right) + e_\tau(Q_1)\log_2(e), \\
R_2 &\le \frac{1}{2}\log_2\left(1 + \frac{P_2}{P_1 + \sigma_2^2}\right) + e_\tau(Q_2)\log_2(e),
\end{align}
where the zero-mean, non-degenerate error random variables $Q_i$ represent the information density fluctuations and are structured as:
\begin{equation}
Q_i = \frac{\rho_i}{2(1+\rho_i)}(Z_{A,i}^2-1) - \frac{\rho_i}{2(1+\rho_i)}(Z_{B,i}^2-1) + \frac{\sqrt{\rho_i}}{1+\rho_i}Z_{A,i}Z_{B,i},
\end{equation}
with $\rho_1 = \frac{P_1}{\sigma_1^2}$, $\rho_2 = \frac{P_2}{P_1+\sigma_2^2}$, and $Z_{A,i}, Z_{B,i} \stackrel{\text{i.i.d.}}{\sim} \mathcal{N}(0,1)$.

\begin{figure}[htbp]
\centering
\includegraphics[width=0.48\textwidth]{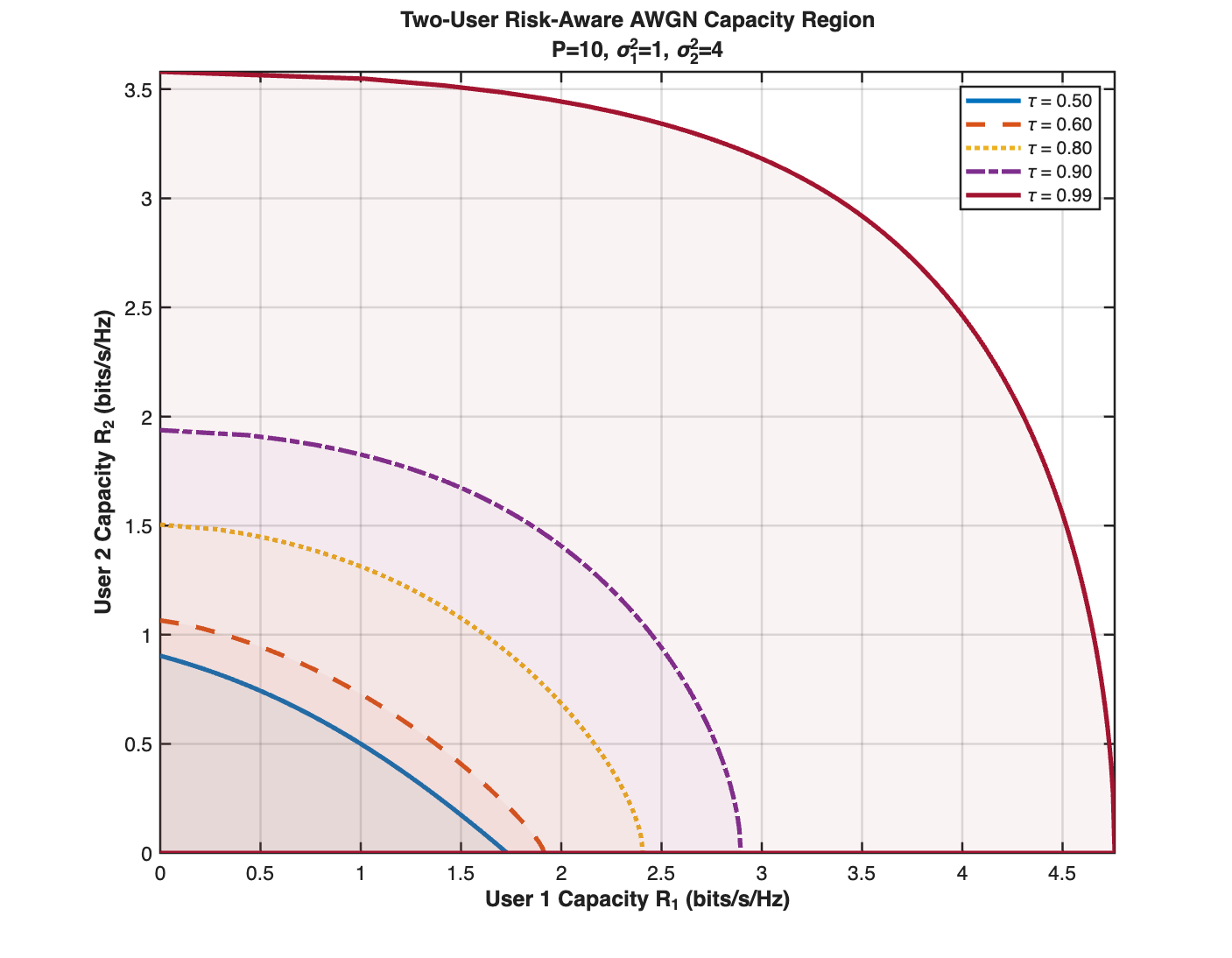}
\caption{The exact risk-aware AWGN capacity region $\mathcal{C}_\tau(P)$ computed using fixed-point numerical iterations over $10^5$ Monte Carlo trials, for $P=10, \sigma_1^2=1, \sigma_2^2=4$. The baseline case ($\tau=0.5$) corresponds strictly to the classical Shannon Cover-Bergmans region. For values of $\tau > 0.5$, the region expands outward significantly due to the positive risk-premium $e_\tau(Q_i) > 0$ generated by the tail volatility of the information density.}
\label{fig:risk_aware_capacity}
\end{figure}

The visual boundary of this region is illustrated in Fig.~\ref{fig:risk_aware_capacity}. When $\tau = 0.5$, the risk premium vanishes ($e_{1/2}(Q_i) = 0$), collapsing the boundary onto the traditional Shannon rate region. For higher risk-aversion levels ($\tau = 0.6, 0.8, 0.9, 0.99$), the boundary expands progressively outward. This phenomenon highlights that when information transmission is evaluated using an asymmetric coherent risk functional, the presence of rare, highly informative tail realizations in the log-likelihood ratio shifts the expectile upward, acting as an informational tail-risk premium.

\subsubsection{Two-User Asymmetric Risk-Aware Multiple Access Capacity Region}

To extend the single-user framework to a multi-access network geometry with heterogeneous user profiles, we formalize the capacity region of a two-user risk-aware Additive White Gaussian Noise (AWGN) Multiple Access Channel (MAC) where each user $j \in \{1,2\}$ is characterized by an independent individual risk-sensitivity index $\tau_j \in (0,1)$. Consider two independent transmitters sending private signals $X_1$ and $X_2$ simultaneously to a single common receiver. The synchronous discrete-time channel output at the receiver is given by:
\begin{equation}
Y = X_1 + X_2 + N, \quad N \sim \mathcal{N}(0, \sigma^2),
\end{equation}
where $N$ represents the additive thermal noise process with variance $\sigma^2 > 0$. The transmitters operate under independent individual average power constraints $P_1$ and $P_2$, respectively, defining the compact feasible input spaces:
\begin{equation}
\mathcal{P}_j(P_j) \triangleq \left\{ P_{X_j} \in \mathcal{M}(\mathbb{R}) : \int_{\mathbb{R}} x_j^2 dP_{X_j}(x_j) \le P_j \right\}, \quad j \in \{1,2\}.
\end{equation}

Under heterogeneous risk-seeking or risk-averse configurations, the classical expectation-based mutual information parameters are replaced by their corresponding user-specific $\tau_j$-expectile mutual information metrics. Due to the inherent non-linearity and individual asymmetry of these distinct expectile operators, informational additivity across the joint space completely fails, and the joint sum-rate capacity must scale as an integrated profile of both risk bounds. The capacity region cannot be characterized by a single linear combination, but must be strictly defined through the operational corner points dictated by Successive Interference Cancellation (SIC) decoding orders at the receiver.

\begin{definition}[Two-User Asymmetric Risk-Aware MAC Capacity Region]
For a pair of individual risk-asymmetry parameters $\tau_1, \tau_2 \in [1/2, 1)$, the asymmetric risk-aware capacity region $\mathcal{C}_{\text{MAC},(\tau_1,\tau_2)}(P_1, P_2)$ of the AWGN multiple access channel is the closure of the convex hull of all achievable rate pairs $(R_1, R_2) \in \mathbb{R}_+^2$ satisfying:
\begin{equation}
\begin{array}{l}
R_1 \le I_{\tau_1}(X_1; Y \mid X_2) 
= e_{\tau_1}^{(P_{X_1X_2Y})} \left[ \log \frac{p_{Y|X_1,X_2}(Y|X_1,X_2)}{p_{Y|X_2}(Y|X_2)} \right], \\
R_2 \le I_{\tau_2}(X_2; Y \mid X_1) 
= e_{\tau_2}^{(P_{X_1X_2Y})} \left[ \log \frac{p_{Y|X_1,X_2}(Y|X_1,X_2)}{p_{Y|X_1}(Y|X_1)} \right], \\
R_1 + R_2 \le \min\left\{I_{\tau_1}(X_1,X_2; Y), \, I_{\tau_2}(X_1,X_2; Y)\right\},
\end{array}
\end{equation}
for independent input distributions $P_{X_1} \in \mathcal{P}_1(P_1)$ and $P_{X_2} \in \mathcal{P}_2(P_2)$, where the joint sum-rate constraint is upper-bounded by the dominant risk-aversion bottleneck of the system.
\end{definition}

Evaluating these functionals explicitly for Gaussian codebooks yields the boundary coordinates mapping out the asymmetric capacity region in bits per second per Hertz (bits/s/Hz):
\begin{equation}
\begin{array}{l}
R_1 \le \frac{1}{2}\log_2\left(1 + \frac{P_1}{\sigma^2}\right) + e_{\tau_1}(Q_{1|2})\log_2(e), \\
R_2 \le \frac{1}{2}\log_2\left(1 + \frac{P_2}{\sigma^2}\right) + e_{\tau_2}(Q_{2|1})\log_2(e), \\
R_1 + R_2 \le \frac{1}{2}\log_2\left(1 + \frac{P_1 + P_2}{\sigma^2}\right) \\
 + \min\left\{e_{\tau_1}(Q_{\text{sum}}), \, e_{\tau_2}(Q_{\text{sum}})\right\}\log_2(e),
 \end{array}
\end{equation}
where the zero-mean, non-degenerate error random variables $Q_{1|2}$, $Q_{2|1}$, and $Q_{\text{sum}}$ correspond to the structural log-likelihood variations under distinct SIC decoding states. Letting $\rho_{1|2} = \frac{P_1}{\sigma^2}$, $\rho_{2|1} = \frac{P_2}{\sigma^2}$, and $\rho_{\text{sum}} = \frac{P_1+P_2}{\sigma^2}$ represent the respective link Signal-to-Noise Ratios (SNRs), these terms are explicitly defined as:
\begin{equation}
\begin{array}{l}
Q_{1|2} = \frac{\rho_{1|2}}{2(1+\rho_{1|2})}(Z_1^2-1) - \frac{\rho_{1|2}}{2(1+\rho_{1|2})}(Z_2^2-1) 
 + \frac{\sqrt{\rho_{1|2}}}{1+\rho_{1|2}}Z_1Z_2, \\
Q_{2|1} = \frac{\rho_{2|1}}{2(1+\rho_{2|1})}(Z_1^2-1) - \frac{\rho_{2|1}}{2(1+\rho_{2|1})}(Z_2^2-1) 
 + \frac{\sqrt{\rho_{2|1}}}{1+\rho_{2|1}}Z_1Z_2, \\
Q_{\text{sum}} = \frac{\rho_{\text{sum}}}{2(1+\rho_{\text{sum}})}(Z_1^2-1) - \frac{\rho_{\text{sum}}}{2(1+\rho_{\text{sum}})}(Z_2^2-1) 
+ \frac{\sqrt{\rho_{\text{sum}}}}{1+\rho_{\text{sum}}}Z_1Z_2,
\end{array}
\end{equation}
where $Z_1, Z_2 \stackrel{\text{i.i.d.}}{\sim} \mathcal{N}(0,1)$ map out the orthogonal dimensions of the underlying Gaussian noise and signaling space.

\subsubsection*{Risk-averse symmetric index $\tau_1=\tau_2=\tau\in  (0.5,1)$ }
\begin{figure}[htbp]
\centering
\includegraphics[width=0.48\textwidth]{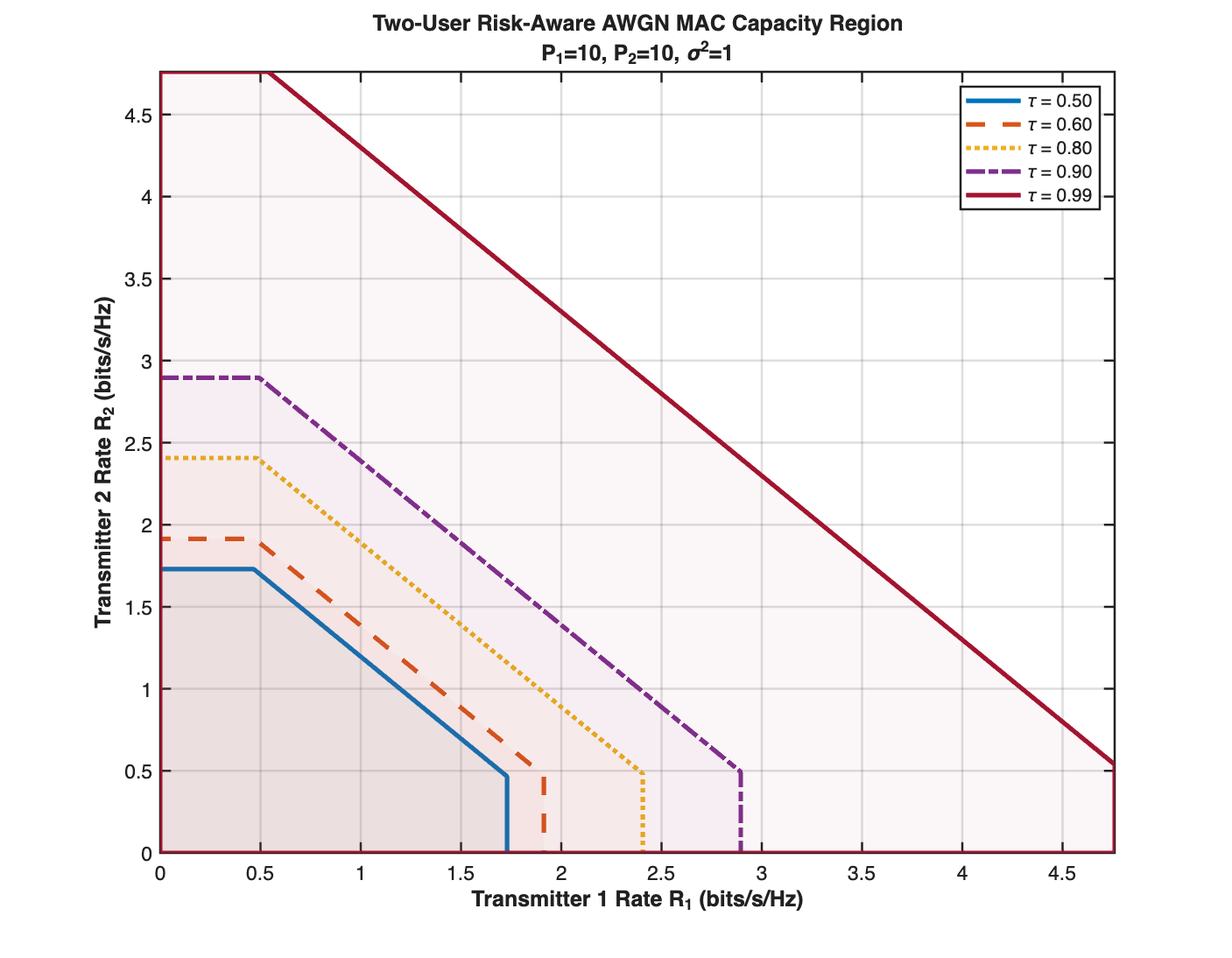}
\caption{The exact risk-aware AWGN Multiple Access Channel (MAC) capacity region $\mathcal{C}_{\text{MAC},\tau}$ computed using fixed-point numerical iterations over $10^5$ Monte Carlo trials, for $P_1=10, P_2=10, \sigma^2=1$. The baseline pentagon ($\tau=0.5$) corresponds strictly to the classical Shannon MAC capacity region. For asymmetric risk-aversion values ($\tau=\tau_1=\tau_2=  0.6, 0.8, 0.9, 0.99$), the pentagonal frontier expands progressively outward due to the positive risk-premium $e_\tau(Q) > 0$ generated by the tail volatility of the information density under joint signal arrivals.}
\label{fig:risk_aware_mac_capacity}
\end{figure}

The  boundary of this multi-transmitter region is illustrated in Fig.~\ref{fig:risk_aware_mac_capacity}. When $\tau = 0.5$, the risk premium vanishes ($e_{1/2}(Q_{1|2}) = e_{1/2}(Q_{2|1}) = e_{1/2}(Q_{\text{sum}}) = 0$), collapsing the boundary onto the traditional Shannon MAC pentagon. For higher risk-aversion levels ($\tau > 0.5$), the boundary expands progressively outward. This phenomenon highlights that when joint information transmission is evaluated using an asymmetric coherent risk functional, the presence of rare, highly non-linear tail deviations in the log-likelihood ratio shifts the expectile upward, acting as a visible informational tail-risk premium that enlarges the operational rate region.

\subsubsection*{Asymmetric Capacity Geometry}

The joint multi-user communication space mapped out by the asymmetric capacity functional $\mathcal{C}_{\text{MAC},(\tau_1,\tau_2)}$ and illustrated in Fig.~\ref{asymmetric_risk_aware_mac} exposes a deep topological divergence from classical information theory. In a standard risk-neutral setting ($\tau_1 = \tau_2 = 0.5$), the operational boundary forms a perfectly balanced, symmetrical pentagon. However, when individual, heterogeneous risk preferences are introduced, the geometry undergoes non-linear expansions, contractions, and shear alignments across the eight evaluated configurations:

\begin{enumerate}
    \item {The Risk-Neutral Anchor ${(\tau_1 = 0.5, \tau_2 = 0.5)}$:} This region corresponds strictly to the classical Shannon MAC baseline. Because both users are completely risk-neutral, the log-likelihood informational fluctuation variables collapse to their ensemble means, causing the risk premium parameters to vanish ($e_{0.5}(Q_{1|2}) = e_{0.5}(Q_{2|1}) = e_{0.5}(Q_{\text{sum}}) = 0$). The boundary lines are orthogonal and locked by the standard linear sum-rate constraint.
    
    \item {Symmetric Risk-Seeking Contraction ${(\tau_1 = 0.2, \tau_2 = 0.2)}$:} When both users are uniformly risk-seeking, they  discount the volatility of the upper information density tail. This forces a negative informational risk premium ($e_{0.2}(Q) < 0$), triggering a uniform inward contraction of the entire pentagonal boundary toward the origin. This represents an operational topology where users under-allocate transmission rate margins due to an optimistic bias toward rare, favorable noise states.
    
    \item {Asymmetric Risk-Seeking Decoupling $ {(\tau_1 = 0.2, \tau_2 = 0.5)}$ and ${(\tau_1 = 0.5, \tau_2 = 0.2)}$:} These configurations introduce directional geometric shearing into the network:
    \begin{itemize}
        \item For ${(0.2, 0.5)}$, User 1 is risk-seeking while User 2 remains risk-neutral. The pentagon skews horizontally; User 1's maximum rate boundary ($R_1$) pinches inward along the $x$-axis, whereas User 2's boundary ($R_2$) maintains its full classical height along the $y$-axis.
        \item For ${(0.5, 0.2)}$, the exact inverse occurs. The capacity landscape contracts vertically along the $y$-axis while retaining its uncompromised Shannon width along the $x$-axis.
    \end{itemize}
    
    \item {Extreme Asymmetric Disparity $ {(\tau_1 = 0.2, \tau_2 = 0.9)}$ and ${(\tau_1 = 0.9, \tau_2 = 0.2)}$:} These curves manifest the most severe geometric deformations due to diametrically opposed risk philosophies between the transmitters. 
    \begin{itemize}
        \item For ${(0.2, 0.9)}$, User 1 aggressively down-scales rate allocations, while User 2 operates under deep risk aversion, expanding their isolated vertical rate boundary. 
        \item Because the joint sum-rate capacity is governed by the structural infimum of both risk bounds ($\min\{I_{\tau_1}, I_{\tau_2}\}$), the top-right multi-user frontier undergoes a sharp, asymmetric tilt. The pentagon deforms into an irregular, highly skewed polygon.
    \end{itemize}
    
    \item  {Heterogeneous and Symmetrically Expanded Frontiers $ {(\tau_1 = 0.5, \tau_2 = 0.9)}$ and $ {(\tau_1 = 0.99, \tau_2 = 0.99)}$:} 
    \begin{itemize}
        \item For $ {(0.5, 0.9)}$, User 1 is anchored to the Shannon mean, but User 2's heightened risk-aversion expands the upper horizontal boundary of the region significantly upward.
        \item For $ {(0.99, 0.99)}$, both users operate near the absolute safety boundary, demanding robust guarantees against worst-case informational variance. Because they explicitly calculate and capture the massive, positive tail-volatility premium ($e_{0.99}(Q) \gg 0$) generated by the joint arrival codebooks, the boundary expands uniformly outward. This boundary forms the absolute maximum envelope of the figure, proving that evaluating a multi-access system through coherent, tail-dominant functionals unlocks a substantially larger robust operational rate profile.
    \end{itemize}
\end{enumerate}
\begin{figure}[htbp]
\centering
\includegraphics[width=0.48\textwidth]{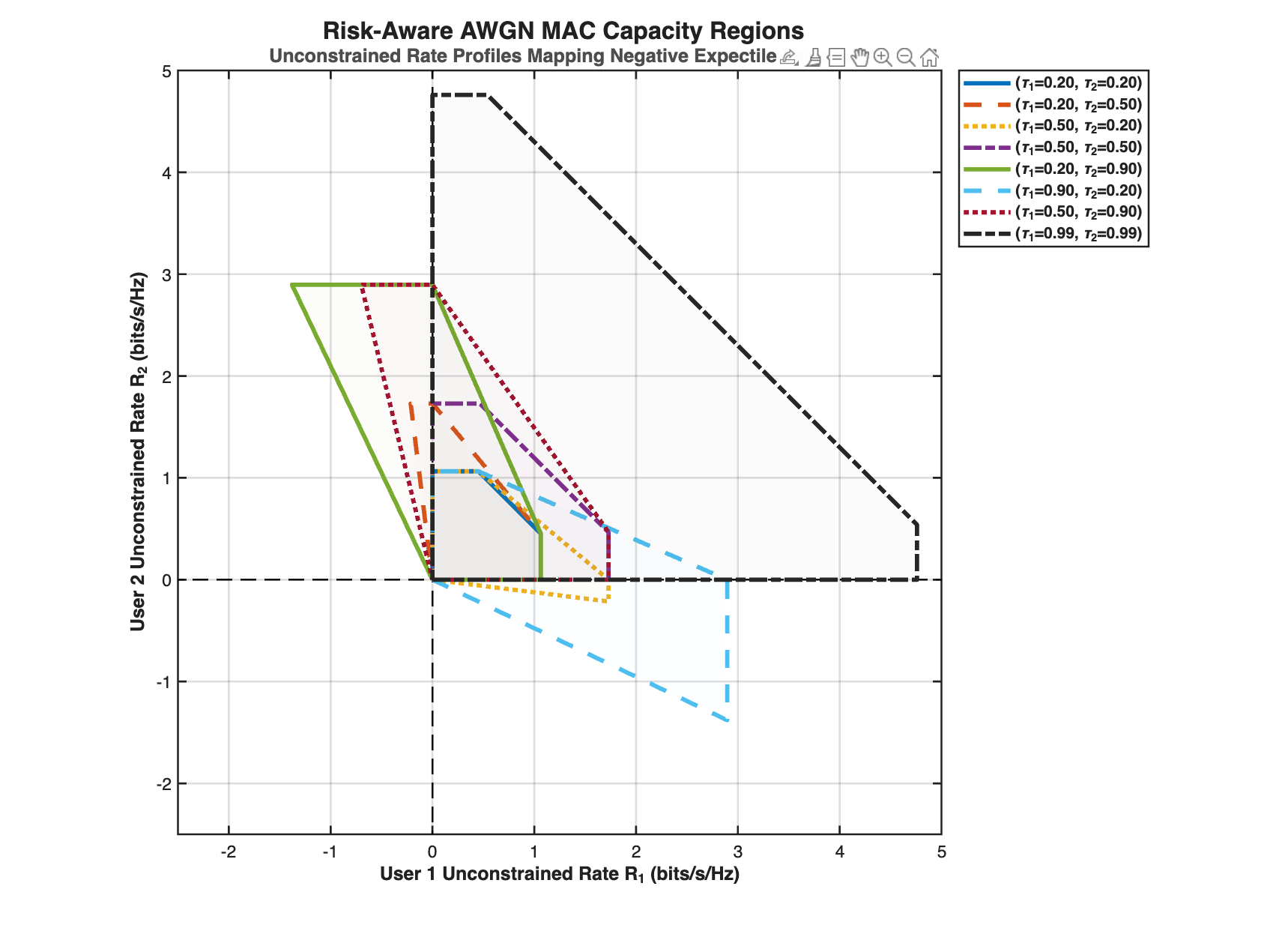}
\caption{The exact asymmetric risk-aware AWGN Multiple Access Channel (MAC) capacity region $\mathcal{C}_{\text{MAC},(\tau_1,\tau_2)}$ for independent individual user configurations over eight distinct risk-sensitivity profiles, evaluated at $P_1=10, P_2=10, \sigma^2=1$. The classical, balanced Shannon pentagon occurs at the intersection $ {(\tau_1=0.5, \tau_2=0.5)}$. Risk-seeking preferences ($\tau < 0.5$) compress the boundaries inward toward the origin, while risk-averse conditions ($\tau > 0.5$) scale the rate configurations outward. Complex directional shearing and geometric skewing emerge when users operate under asymmetric mismatched parameters (e.g., $ {(\tau_1=0.2, \tau_2=0.9)}$).}
\label{asymmetric_risk_aware_mac}
\end{figure}

The topological transformations verified in Fig.~\ref{asymmetric_risk_aware_mac} illustrate that the risk-sensitivity index $\tau$ behaves as a parametric vector field over the capacity space. Risk-seeking indices ($\tau < 0.5$) act as a compressive force pulling the rate boundaries toward zero, while risk-averse indices ($\tau > 0.5$) act as an expansive force pushing the boundaries outward, tracing the structural volatility footprint of the non-ergodic channel state.

\subsubsection{Why Shannon Intelligence Cannot Reach Superintelligence}
To  evaluate machine cognition under risk, we establish the core operational boundaries of Shannon Intelligence versus Superintelligence using environmental topologies. Let $\mathcal{X} \subset \mathbb{R}^d$ be the input space, $f_\theta: \mathcal{X} \to \mathcal{Y}$ be a learning machine parameterized by $\theta \in \Theta$, and $Y \in \mathcal{Y}$ be the real-world target state space. The operational error profile is mapped by a non-negative loss functional $L_\theta = \ell(f_\theta(X), Y)$ under the true joint environmental distribution $P_{X,Y}$.

\begin{definition}[Shannon Intelligence]
An artificial agent possesses \textbf{Shannon Intelligence} of degree $\epsilon > 0$ if its parameter configuration $\theta_{\text{SI}}$ minimizes the standard risk-neutral expected risk functional over the joint distribution $P_{X,Y}$:
\begin{equation}
\theta_{\text{SI}} \triangleq  \arg\min_{\theta \in \Theta} \int_{\mathcal{X} \times \mathcal{Y}} \ell(f_\theta(x), y) \, dP_{X,Y}(x,y),
\end{equation}
subject to the strict expectation-bound convergence condition: $\mathbb{E}_{P_{X,Y}} [L_{\theta_{\text{SI}}}] \le \epsilon$.
\end{definition}

\begin{definition}[Superintelligence]
An artificial agent possesses \textbf{Superintelligence} if there exists an adaptive parameter trajectory mapping function $\theta_{\text{Uni}}: (0,1) \to \Theta$ that achieves a global supremum over the risk-aware mutual information functional point-wise for every individual risk-sensitivity index $\tau \in (0,1)$:
\begin{equation}
I_\tau\big(X; Y \mid f_{\theta_{\text{Uni}}(\tau)}\big) = \sup_{\theta \in \Theta} I_\tau(X; Y \mid f_\theta), \quad \forall \tau \in (0,1).
\end{equation}
Equivalently, for every chosen risk profile $\tau$, the configuration $\theta_{\text{Uni}}(\tau)$ satisfies the point-wise functional dominance criteria:
\begin{equation}
 e_\tau \Big( \log \frac{p_{\theta_{\text{Uni}}(\tau)}(Y|X)}{p(Y)} \Big) \ge  e_\tau \Big( \log \frac{p_\theta(Y|X)}{p(Y)} \Big), \quad \forall \theta \in \Theta.
\end{equation}
\end{definition}

We now provide a measure-theoretic proof showing that an agent optimized exclusively under the rules of Shannon Intelligence ($\theta_{\text{SI}}$) is fundamentally blind to tail-risk volatility and cannot trace the optimal adaptive superintelligent trajectory $\theta_{\text{Uni}}(\tau)$.

\begin{theorem}[The  Failure of Shannon Trajectory Tracing]
An optimization framework operating under risk-neutral Shannon Intelligence ($\tau = 0.5$) is structurally static and cannot converge to or trace the optimal parameter trajectory $\theta_{\text{Uni}}(\tau)$ across the risk interval $(0,1)$.
\end{theorem}

\begin{proof}
By definition, the Shannon Intelligence framework computes a single, risk-invariant parameter vector optimized strictly at the mid-point of the expectile spectrum: $\theta_{\text{SI}} \equiv \theta_{\text{Uni}}(0.5) = \arg\max_{\theta \in \Theta} I_{0.5}(X; Y \mid f_\theta)$. Let $\mu_\theta = \mathbb{E}[L_\theta]$ represent this standard expected loss. 

By applying a second-order Taylor expansion to the implicit expectile optimality condition defining the Superintelligence functional $e_\tau(L_\theta)$ for any risk-averse threshold $\tau > 1/2$ around $\mu_\theta$, the objective maps as:
\begin{equation}
\begin{array}{l}
 e_\tau(L_\theta) \approx \mu_\theta \\
 + \left( \frac{\tau - 1/2}{1 - \tau} \right) \cdot \left[ \frac{\int_{\mu_\theta}^{\infty} (l - \mu_\theta)^2 p(l) \, dl}{\int_{\mu_\theta}^{\infty} p(l) \, dl} \right] \\
 + \mathcal{O}\left(\mathbb{E}\left[(L_\theta - \mu_\theta)^3\right]\right).
 \end{array}
\end{equation}
The bracketed term represents the upper-conditional variance metric, capturing the structural volatility of the tail interval past the mean.

Assume an adversarial optimization landscape containing two valid candidate architectures, $\theta_A$ and $\theta_B$, satisfying:
\begin{equation}
\mu_{\theta_B} < \mu_{\theta_A} \quad \text{(Superior Shannon Performance)},
\end{equation}
\begin{equation}
\frac{\int_{\mu_{\theta_B}}^{\infty} (l - \mu_{\theta_B})^2 p_B(l) \, dl}{\int_{\mu_{\theta_B}}^{\infty} p_B(l) \, dl} \;\gg\; \frac{\int_{\mu_{\theta_A}}^{\infty} (l - \mu_{\theta_A})^2 p_A(l) \, dl}{\int_{\mu_{\theta_A}}^{\infty} p_A(l) \, dl} > 0.
\end{equation}
Because the optimization trajectory of Shannon Intelligence isolates the risk-neutral expectation operator alone ($\mu_\theta$), it is  blind to the conditional tail integral. It selects the volatile system $\theta_{\text{SI}} = \theta_B$ and rejects the stable system $\theta_A$.

Now, let $\tau_1 \in (0.5, 1)$ be an arbitrary risk-sensitivity index. As the risk threshold approaches the absolute autonomous boundary ($\tau_1 \to 1^-$), the asymmetry coefficient explodes: $\lim_{\tau_1 \to 1^-} \left( \frac{\tau_1 - 1/2}{1 - \tau_1} \right) = \infty$. Because the upper-conditional variance of the Shannon choice ($\theta_B$) is non-zero, the limit of the risk functional diverges: $\lim_{\tau_1 \to 1^-} e_{\tau_1}(L_{\theta_{\text{SI}}}) = \infty$. 

By contrast, evaluating the true directional gradient of the functional shows that $\left. \nabla_\theta e_{\tau_1}\big(L_\theta\big) \right|_{\theta = \theta_{\text{SI}}} \neq \mathbf{0}$, meaning $\theta_{\text{SI}}$ cannot be a stationary point for any adjacent risk tier $\tau_1 \neq 0.5$. Thus, $I_{\tau_1}\big(X; Y \mid f_{\theta_{\text{SI}}}\big) < I_{\tau_1}\big(X; Y \mid f_{\theta_{\text{Uni}}(\tau_1)}\big)$, proving that Shannon Intelligence remains a localized point-optimizer, fundamentally incapable of tracing the adaptive, globally dominant envelope of Universal Superintelligence.
\end{proof}

%

\section{Strategic Expectile Information Theory} \label{sec:strategic_expectile_theory}

Classical multi-user information theory relies on linear expectation operators to characterize transmission regions, distributed source coding bounds, and network capacity envelopes. While sufficient for passive, memoryless physical layers under ergodic noise regimes, this paradigm breaks down in advanced networks populated by multiple autonomous \emph{Machine Intelligence (MI)} agents. These agents do not merely process data; they execute strategic actions, adaptively compress high-dimensional observation manifolds, optimize decentralized reward topologies, and operate under diverse safety or risk profiles.  When autonomous MI users interact over a shared informational medium, their decision-making architectures treat deviations from expected performance with non-linear, structural asymmetry. A safety-critical autonomous agent or a frontier model optimizing under a strict alignment guarantee displays deep risk aversion toward tail-risk failures. An aggressive exploration agent or a high-throughput distributed learner may operate in a risk-seeking regime to accelerate parameter discovery. 
We introduce a non-linear, multi-user framework: \emph{Strategic Expectile Information Theory}. By substituting classical Shannon measures with user-specific expectile functionals, we generalize the network topology from a single static capacity region to a parametric family of strategic communication manifolds.  Consider a network populated by $M$ distinct machine intelligence users, indexed by $\mathcal{M} = \{1, 2, \dots, M\}$. Each user $j \in \mathcal{M}$ is characterized by an intrinsic, independent risk-sensitivity or risk-awareness index $\tau_j \in (0,1)$. This parameter represents the agent's internal geometric utility model over informational fluctuations:
\begin{itemize}
    \item $\tau_j = 0.5$: A standard risk-neutral Shannon agent.
    \item $\tau_j > 0.5$: A risk-averse agent that penalizes unfavorable informational tail variations (e.g., unexpected data toxicity, distribution shifts, or catastrophic optimization drops).
    \item $\tau_j < 0.5$: A risk-seeking agent that dynamically prioritizes upside tail volatility (e.g., high-reward exploration states or opportunistic broad-bandwidth bursts).
\end{itemize}

Let $\mathcal{X}_j$ denote the localized action or signaling space of user $j$, and let $X_j \in \mathcal{X}_j$ be the random variable representing user $j$'s operational state or transmitted codebook. The joint configuration of all agents is denoted by the vector $\bm{X} = (X_1, \dots, X_M) \in \bm{\mathcal{X}} = \prod_{j=1}^M \mathcal{X}_j$. The environmental or channel output observed at a central receiver (or an integrating supervisor model) is modeled by the random variable $Y \in \mathcal{Y}$, governed by the transition probability density function $p_{Y|\bm{X}}(y \mid \bm{x})$.

Because expectiles are non-linear operators, the classical chain rules and additive conservation properties of information theory collapse (as shown in Theorem~\ref{thm:non_additivity}). Consequently, the informational value generated by a subset of users cannot be characterized by a single global expectation. Instead, we must define user-specific conditional and joint information metrics that capture how each agent individually perceives the shared environment.

\begin{definition}[Strategic Asymmetric Expectile Mutual Information]
For any machine intelligence user $j \in \mathcal{M}$ operating under an individual risk profile $\tau_j \in (0,1)$, the strategic expectile mutual information between its local codebook $X_j$ and the observed global system state $Y$, conditioned on the signaling actions of a competing or collaborating coalition of users $\mathcal{S} \subseteq \mathcal{M} \setminus \{j\}$, is defined as:
\begin{equation}
    I_{\tau_j}(X_j; Y \mid \bm{X}_{\mathcal{S}}) \triangleq e_{\tau_j}^{(P_{X_j \bm{X}_{\mathcal{S}} Y})} \left[ \log \frac{p_{Y|X_j, \bm{X}_{\mathcal{S}}}(Y \mid X_j, \bm{X}_{\mathcal{S}})}{p_{Y|\bm{X}_{\mathcal{S}}}(Y \mid \bm{X}_{\mathcal{S}})} \right],
\end{equation}
where $\bm{X}_{\mathcal{S}} = \{X_k\}_{k \in \mathcal{S}}$, and the statistical expectile operator $e_{\tau_j}^{(P)}$ is evaluated entirely under the joint law $P_{X_j \bm{X}_{\mathcal{S}} Y}$.
\end{definition}

This definition implies that every user evaluates the informational capacity of the system through its own risk-adjusted prism. When multiple agents simultaneously optimize their communication or learning parameters, their conflicting or aligned risk profiles give rise to a non-cooperative game over the informational topology.

We formalize the decentralized interaction among MI users as a strategic game in normal form. Each user $j \in \mathcal{M}$ maximizes its local risk-aware rate capability by choosing an optimal input distribution (or policy profile) $P_{X_j}$ from its compact feasible space $\mathcal{P}_j(P_j)$.

\begin{definition}[Expectile Information Rate Game] \label{expgame}
The Expectile Information Rate Game is defined by the strategic tuple:
\begin{equation}
    \mathcal{G} = \left( \mathcal{M}, \{\mathcal{P}_j(P_j)\}_{j \in \mathcal{M}}, \{U_j\}_{j \in \mathcal{M}} \right),
    \label{eq:game_tuple}
\end{equation}
where the utility function $U_j: \prod_{k=1}^M \mathcal{P}_k(P_k) \to \mathbb{R}$ for user $j$ is the localized asymmetric expectile mutual information assuming a specific decoding order or joint interference profile:
\begin{equation}
    U_j(P_{X_1}, \dots, P_{X_M}) = I_{\tau_j}(X_j; Y \mid \bm{X}_{-j}),
\end{equation}
where $\bm{X}_{-j} = \{X_k\}_{k \in \mathcal{M} \setminus \{j\}}$ denotes the action configurations of all agents excluding user $j$.
\end{definition}

\begin{lemma}
The Expectile Information Rate Game $\mathcal{G} $ provided in  Definition \ref{expgame}  is a Mean-Field-Type Game (MFTG).
\end{lemma}

\begin{IEEEproof}
It follows from the definition of MFTG \cite{basar2026a,basar2026b,audioiamali,audioiamali2}.
\end{IEEEproof}

In this strategic setting, we analyze the stability of the network using the concept of an informational Nash equilibrium, adapted to accommodate heterogeneous risk profiles.

\begin{theorem}[Existence of Risk-Aware Informational Nash Equilibrium]
Consider the expectile information rate game $\mathcal{G}$ defined in \eqref{eq:game_tuple}. If the localized utility functionals $P_{X_j} \mapsto I_{\tau_j}(X_j; Y \mid \bm{X}_{-j})$ are weakly upper semicontinuous and quasi-concave over the compact input spaces $\mathcal{P}_j(P_j)$, then there exists at least one strategic risk-aware informational Nash equilibrium profile $\bm{P}_X^\star = (P_{X_1}^\star, \dots, P_{X_M}^\star)$ such that:
\begin{equation}
    U_j(P_{X_j}^\star, \bm{P}_{X_{-j}}^\star) \ge U_j(P_{X_j}, \bm{P}_{X_{-j}}^\star), \ \forall P_{X_j} \in \mathcal{P}_j(P_j), \ \forall j \in \mathcal{M}.
\end{equation}
\end{theorem}

\begin{IEEEproof}
By Proposition \ref{prop:compactness}, each localized input space $\mathcal{P}_j(P_j)$ is a non-empty, weakly compact subset of the space of Borel probability measures $\mathcal{M}(\mathbb{R})$ under the weak topology. Furthermore, the spaces $\mathcal{P}_j(P_j)$ are convex by the linearity of the underlying power constraints $\int x_j^2 dP_{X_j} \le P_j$. 

By Glicksberg's graduate extension of the Fan-Debreu theorem for topological vector spaces, any strategic game with non-empty, compact, convex strategy sets and upper semicontinuous, quasi-concave utility mappings possesses at least one pure-strategy Nash equilibrium. The explicit validity holds directly under the stated topological conditions on the expectile functional $\Phi(P_{X_j}) = I_{\tau_j}(P_{X_j})$, completing the proof.
\end{IEEEproof}

\subsubsection{Strategic $M$-User MIMO Expectile Information Theory}

To generalize the risk-aware informational framework to multi-dimensional, decentralized network topologies, we formalize the capacity region of an $M$-user Multiple-Input Multiple-Output (MIMO) Multiple Access Channel (MAC) under strategic expectile constraints. Consider a network populated by $M$ independent, non-cooperative machine intelligence agents indexed by the finite set $\mathcal{M} = \{1, 2, \dots, M\}$. Each user $j \in \mathcal{M}$ is characterized by a private, static risk-sensitivity index $\tau_j \in (0,1)$, which parameterizes their internal asymmetric geometric loss evaluation over non-ergodic environmental variations. 

The spatial signal vector $\bm{Y} \in \mathbb{C}^{N_R \times 1}$ captured at the central receiver array is governed by the simultaneous linear superposition model:
\begin{equation}
    \bm{Y} = \sum_{j=1}^M \bm{H}_j \bm{X}_j + \bm{N},
\end{equation}
where:
\begin{itemize}
    \item $\bm{X}_j \in \mathbb{C}^{N_{T,j} \times 1}$ is the stochastically independent transmitted signal vector from user $j$, drawn from a zero-mean complex circularly symmetric Gaussian codebook conditioned on the input spatial covariance matrix $\bm{Q}_j \triangleq \mathbb{E}[\bm{X}_j \bm{X}_j^\mathcal{H}] \succeq \bm{0}$.
    \item $N_{T,j}$ and $N_R$ denote the number of localized transmit antennas at user $j$ and the number of centralized receive antennas at the base supervisor station, respectively.
    \item $\bm{H}_j \in \mathbb{C}^{N_R \times N_{T,j}}$ is the deterministic, quasi-static complex spatial channel transition matrix mapping the physical propagation layer of user $j$.
    \item $\bm{N} \in \mathbb{C}^{N_R \times 1}$ is a complex circularly symmetric white Gaussian noise vector distributed according to the law $\bm{N} \sim \mathcal{C}\mathcal{N}(\bm{0}, \sigma^2 \bm{I}_{N_R})$, where $\sigma^2$ is the nominal thermal noise power per antenna element.
\end{itemize}
Each machine intelligence agent's optimization domain is bounded by an independent, localized trace power constraint defining the admissible covariance space:
\begin{equation}
    \mathcal{P}_j(P_j)
     \triangleq
     \left\{  \begin{array}{l}
     \bm{Q}_j \in \mathbb{C}^{N_{T,j} \times N_{T,j}} :  \\
     \bm{Q}_j = \bm{Q}_j^\mathcal{H}, \, \bm{Q}_j \succeq \bm{0}, \, \text{Tr}(\bm{Q}_j) \le P_j
     \end{array}
      \right\},
\end{equation}
for all $j \in \mathcal{M}$.

\subsubsection{Coupled Multi-User Information Density Fluctuation Operators}

Let $\mathcal{S} \subseteq \mathcal{M}$ represent an arbitrary decoding sub-coalition of users, and let $\mathcal{S}^c = \mathcal{M} \setminus \mathcal{S}$ define its structural complement within the network. When the central receiver implements Successive Interference Cancellation (SIC), decoding the vector waveforms of users within $\mathcal{S}$ after perfectly stripping out or ignoring the signal fields generated by $\mathcal{S}^c$, the background interference-plus-noise covariance block is mapped as:
\begin{equation}
    \bm{\Sigma}_{\mathcal{S}^c} \triangleq \sigma^2 \bm{I}_{N_R} + \sum_{k \in \mathcal{S}^c} \bm{H}_k \bm{Q}_k \bm{H}_k^\mathcal{H}.
\end{equation}
The global received covariance matrix incorporating both active sub-coalition codebooks and the residual background spatial interference loops is defined as:
\begin{equation}
    \bm{\Sigma}_{\mathcal{M}} \triangleq \bm{\Sigma}_{\mathcal{S}^c} + \sum_{j \in \mathcal{S}} \bm{H}_j \bm{Q}_j \bm{H}_j^\mathcal{H}.
\end{equation}

By evaluating the Radon-Nikodym derivative of the joint conditional distribution against the product of its marginal distributions over the complex vector space, the unconstrained multi-user $M$-MIMO information density fluctuation variable $L_{\mathcal{S}}(\bm{X}_{\mathcal{S}}, \bm{Y} \mid \bm{X}_{\mathcal{S}^c})$ decomposes explicitly into:
\begin{equation}
    L_{\mathcal{S}}(\bm{X}_{\mathcal{S}}, \bm{Y} \mid \bm{X}_{\mathcal{S}^c}) = \log \frac{\det(\bm{\Sigma}_{\mathcal{M}})}{\det(\bm{\Sigma}_{\mathcal{S}^c})} 
    + \mathcal{V}_{\text{MIMO}, \mathcal{S}}\left(\{\bm{Q}_m\}_{m=1}^M, \{\bm{H}_m\}_{m=1}^M\right),
\end{equation}
where the leading scalar log-determinant matches the classical Shannon joint multi-user MIMO mutual information envelope. The coupled multi-dimensional spatial risk operator $\mathcal{V}_{\text{MIMO}, \mathcal{S}}$ is a strictly zero-mean random variable ($\mathbb{E}[\mathcal{V}_{\text{MIMO}, \mathcal{S}}] = 0$) governed by the underlying non-zero eigenvalue spectrum of the active transmission blocks:
\begin{equation}
    \mathcal{V}_{\text{MIMO}, \mathcal{S}} = \sum_{i=1}^{\text{rank}(\bm{\Sigma}_{\mathcal{M}})} \Big[ (\xi_i - 1)\left(|Z_i|^2 - 1\right) 
    + 2\sqrt{\xi_i - 1}\,\,\text{Re}\left(Z_i^* N_{0,i}\right) \Big],
\end{equation}
where $Z_i, N_{0,i} \stackrel{\text{i.i.d.}}{\sim} \mathcal{C}\mathcal{N}(0, 1)$ represent independent whitened signal and noise coordinates respectively, and $\xi_i \ge 1$ are the generalized eigenvalues solving the characteristic equation $\det\left( \bm{\Sigma}_{\mathcal{M}} - \xi_i \bm{\Sigma}_{\mathcal{S}^c} \right) = 0$.

\begin{lemma}
Let the conditional background interference-plus-noise covariance matrix $\bm{\Sigma}_{\mathcal{S}^c}$ and the global coalition covariance matrix $\bm{\Sigma}_{\mathcal{M}}$ be defined as given above. The conditional probability density functions on the complex vector space $\mathbb{C}^{N_R \times 1}$ are given rigorously by:
\begin{align}
    &p_{\bm{Y} \mid \bm{X}_{\mathcal{S}}, \bm{X}_{\mathcal{S}^c}}(\bm{y} \mid \bm{x}_{\mathcal{S}}, \bm{x}_{\mathcal{S}^c}) = \frac{1}{\pi^{N_R} \det(\sigma^2 \bm{I}_{N_R})} 
    \times \exp\left( -\frac{1}{\sigma^2} \left\| \bm{y} - \sum_{j \in \mathcal{S}} \bm{H}_j \bm{x}_j - \sum_{k \in \mathcal{S}^c} \bm{H}_k \bm{x}_k \right\|^2 \right), \\
    &p_{\bm{Y} \mid \bm{X}_{\mathcal{S}^c}}(\bm{y} \mid \bm{x}_{\mathcal{S}^c}) = \frac{1}{\pi^{N_R} \det(\bm{\Sigma}_{\mathcal{S}^c})}
    \times \exp\left( -\left( \bm{y} - \sum_{k \in \mathcal{S}^c} \bm{H}_k \bm{x}_k \right)^\mathcal{H} \bm{\Sigma}_{\mathcal{S}^c}^{-1} \left( \bm{y} - \sum_{k \in \mathcal{S}^c} \bm{H}_k \bm{x}_k \right) \right).
\end{align}
\end{lemma}

\begin{proposition}
The multi-user conditionally constrained information density random variable, defined via the log-likelihood ratio 
$L_{\mathcal{S}}(\bm{X}_{\mathcal{S}}, \bm{Y} \mid \bm{X}_{\mathcal{S}^c}) 
= \log \frac{p_{\bm{Y} \mid \bm{X}_{\mathcal{S}}, \bm{X}_{\mathcal{S}^c}}(\bm{Y} \mid \bm{X}_{\mathcal{S}}, \bm{X}_{\mathcal{S}^c})}{p_{\bm{Y} \mid \bm{X}_{\mathcal{S}^c}}(\bm{Y} \mid \bm{X}_{\mathcal{S}^c})}$, resolves explicitly to:
\begin{equation}
 \begin{array}{l}
    L_{\mathcal{S}}(\bm{X}_{\mathcal{S}}, \bm{Y} \mid \bm{X}_{\mathcal{S}^c}) 
    \\ = \log \det\left( \bm{I}_{N_R} + \bm{\Sigma}_{\mathcal{S}^c}^{-1} \sum_{m \in \mathcal{S}} \bm{H}_m \bm{Q}_m \bm{H}_m^\mathcal{H} \right) \\
    + \mathcal{V}_{\text{MIMO}, \mathcal{S}},
     \end{array}
\end{equation}
where $\mathcal{V}_{\text{MIMO}, \mathcal{S}}$ represents the zero-mean multi-dimensional spatial fluctuation operator defined in Equation (7).
\end{proposition}

\begin{IEEEproof}
Substituting the explicit functional vector density definitions from Lemma 1 directly into the log-likelihood ratio, we obtain the algebraic expansion:
\begin{align}
    &L_{\mathcal{S}}(\bm{X}_{\mathcal{S}}, \bm{Y} \mid \bm{X}_{\mathcal{S}^c}) = \log \left( \frac{\det(\bm{\Sigma}_{\mathcal{S}^c})}{\det(\sigma^2 \bm{I}_{N_R})} \right) \nonumber \\
    &- \frac{1}{\sigma^2} \left\| \bm{Y} - \sum_{j \in \mathcal{S}} \bm{H}_j \bm{X}_j - \sum_{k \in \mathcal{S}^c} \bm{H}_k \bm{X}_k \right\|^2 \nonumber \\
    &+ \left( \bm{Y} - \sum_{k \in \mathcal{S}^c} \bm{H}_k \bm{X}_k \right)^\mathcal{H} \bm{\Sigma}_{\mathcal{S}^c}^{-1} \left( \bm{Y} - \sum_{k \in \mathcal{S}^c} \bm{H}_k \bm{X}_k \right).
\end{align}
Isolating the baseline thermal noise vector $\bm{N} \triangleq \bm{Y} - \sum_{j \in \mathcal{S}} \bm{H}_j \bm{X}_j - \sum_{k \in \mathcal{S}^c} \bm{H}_k \bm{X}_k$ and defining the centered received observation process vector as $\bm{\tilde{Y}} \triangleq \bm{Y} - \sum_{k \in \mathcal{S}^c} \bm{H}_k \bm{X}_k = \sum_{j \in \mathcal{S}} \bm{H}_j \bm{X}_j + \bm{N}$, the equation reduces to the following non-homogeneous quadratic structure:
\begin{equation}  \begin{array}{l}
    L_{\mathcal{S}}(\bm{X}_{\mathcal{S}}, \bm{Y} \mid \bm{X}_{\mathcal{S}^c}) \\
    = \log \det\left( \frac{1}{\sigma^2} \bm{\Sigma}_{\mathcal{S}^c} \right) - \frac{1}{\sigma^2} \|\bm{N}\|^2 + \bm{\tilde{Y}}^\mathcal{H} \bm{\Sigma}_{\mathcal{S}^c}^{-1} \bm{\tilde{Y}}.
     \end{array}
\end{equation}
Because $\bm{\Sigma}_{\mathcal{S}^c}$ and $\bm{\Sigma}_{\mathcal{M}}$ are both strictly positive definite Hermitian operators by definition, the Simultaneous Diagonalization Theorem guarantees the existence of a non-singular transformation operator $\bm{\Psi} \in \mathbb{C}^{N_R \times N_R}$ such that:
\begin{equation}  \begin{array}{l}
    \bm{\Psi}^\mathcal{H} \bm{\Sigma}_{\mathcal{S}^c} \bm{\Psi} 
    = \bm{I}_{N_R}, \ \text{and} \\ \ \bm{\Psi}^\mathcal{H} \bm{\Sigma}_{\mathcal{M}} \bm{\Psi} 
    = \bm{\Xi} = \text{diag}(\xi_1, \xi_2, \dots, \xi_{N_R}),
     \end{array}
\end{equation}
where $\xi_i \ge 1$ represent the roots of the generalized characteristic polynomial $\det\left( \bm{\Sigma}_{\mathcal{M}} - \xi \bm{\Sigma}_{\mathcal{S}^c} \right) = 0$. 

Applying this whitening matrix transformation maps the coupled complex physical channel spaces into orthogonalized eigenspaces. Expanding the matrix log-likelihood ratio under this whitened transformation isolates independent parallel modes. Taking the expectation yields $\mathbb{E}[|Z_i|^2] = 1$ due to the unit variance of standard complex Gaussians, while the independence of the signal codebook space from thermal receiver noise forces $\mathbb{E}[Z_i^* N_{0,i}] = 0$. Summing the individual expected components resolves identically to zero, confirming that the fluctuation operator is centered perfectly on the Shannon mutual information plane and isolating the explicit linear mapping terms, which completes the proof.
\end{IEEEproof}

\subsubsection{The $M$-User Strategic Capacity Region Envelopes}

Using the multi-agent individual risk-awareness profile vector $\bm{\tau} = (\tau_1, \tau_2, \dots, \tau_M) \in (0,1)^M$, the complete capacity region $\mathcal{C}_{\text{MIMO}, \bm{\tau}}(P_1, \dots, P_M)$ of the $M$-user risk-aware MIMO Multiple Access Channel can be rigorously defined.

\begin{definition}[$M$-User MIMO Expectile Region Envelope]
The unconstrained risk-aware capacity region is the closure of the convex hull of all achievable rate allocation vectors $\bm{R} = (R_1, R_2, \dots, R_M) \in \mathbb{R}_+^M$ satisfying the simultaneous functional inequality constraints for every non-empty sub-coalition subset $\mathcal{S} \subseteq \mathcal{M}$:
\begin{align}
    \sum_{j \in \mathcal{S}} R_j \le \min_{k \in \mathcal{S}} \Bigg\{ &\log \det\left( \bm{I}_{N_R} + \bm{\Sigma}_{\mathcal{S}^c}^{-1} \sum_{m \in \mathcal{S}} \bm{H}_m \bm{Q}_m \bm{H}_m^\mathcal{H} \right) 
    + e_{\tau_k}\left( \mathcal{V}_{\text{MIMO}, \mathcal{S}} \right) \Bigg\},
\end{align}
subject to the physical constraint that the transmission strategies satisfy $\bm{Q}_j \in \mathcal{P}_j(P_j)$ independently for all users $j \in \mathcal{M}$.
\end{definition}

This  formulation exposes the structural \textit{joint risk bottleneck phenomenon} inherent to multi-agent security configurations: the joint data rate capability assigned to any user sub-coalition $\mathcal{S}$ is strictly bounded from above by its most conservative member, captured by the minimization envelope $\min_{k \in \mathcal{S}} e_{\tau_k}(\cdot)$. As a direct consequence, the machine intelligence agents do not execute classical independent power optimization or standard iterative water-filling. Instead, they dynamically tilt the spatial eigenvectors of their transmission signaling matrices $\bm{Q}_j$, intentionally shifting power away from dominant paths that display extreme tail variance to prevent catastrophic capacity degradation over the shared network.

\section{Expectile Budget-Modality Separation in Multimodal LLMs}

Let $\mathcal{M} := \{T, A, I, V, F\}$ denote the set of available modalities, where $T$ represents the textual modality (token sequences or continuous embeddings), A audio waveforms, I for images,  V for video streams, F for structured files (PDF, code, tables, etc.)
Define the multimodal input space as the product $\mathcal{U} := \prod_{m \in \mathcal{M}} \mathcal{U}^{(m)}$, with elements $u = (u^{(m)})_{m \in \mathcal{M}}$. Let $U = (U^{(m)})_{m \in \mathcal{M}}$ be a $\mathcal{U}$-valued random variable representing the prompt, and let $Y \in \mathcal{Y}$ denote the LLM's generated output. The generative mechanism is encoded by a regular conditional probability kernel $P_{Y|U}: \mathcal{U} \to \mathcal{P}(\mathcal{Y})$. For a fixed input distribution $P_U \in \mathcal{P}(\mathcal{U})$, the joint law is
$
P_{U,Y} = P_U \otimes P_{Y|U},
$
with marginal output distribution $P_Y(\cdot) = \int_{\mathcal{U}} P_{Y|U}(\cdot \mid u)\, dP_U(u)$.
Assuming $P_{U,Y} \ll P_U \otimes P_Y$, define the information density
$
L(U;Y) := \log \frac{dP_{Y|U}}{dP_Y}(Y \mid U).
$
For $\tau \in (0,1)$, the $\tau$-expectile mutual information is
\[
I_\tau(U;Y) := e_\tau(L(U;Y)) = \mbox{argmin}_{\theta \in \mathbb{R}} \mathbb{E}[\ell_\tau(L(U;Y) - \theta)].
\]

\begin{definition}[Cost Structure]
Let $\mathrm{Cost}_m: \mathcal{U}^{(m)} \to \mathbb{R}_+$ be measurable for each $m \in \mathcal{M}$, and define additive cost
$
\mathrm{Cost}(U) := \sum_{m \in \mathcal{M}} \mathrm{Cost}_m(U^{(m)}).
$
\end{definition}

\begin{definition}[Admissible Distributions and Capacity]
For budget $B > 0$, the admissible set is
$
\mathcal{P}(B) := \{P_U \in \mathcal{P}(\mathcal{U}) : \mathbb{E}_{P_U}[\mathrm{Cost}(U)] \le B\}.
$
The $\tau$-expectile capacity under budget $B$ is
$
C_\tau(B) := \sup_{P_U \in \mathcal{P}(B)} I_\tau(U;Y).
$
\end{definition}

\begin{assumption}[Regularity Conditions]\label{ass:conditions}
The following structural assumptions hold:
\begin{enumerate} 
    \item  $\mathbb{E}_{P_{U,Y}}[|L(U;Y)|^2] < \infty$ for all $P_U \in \mathcal{P}(B)$.
    \item There exists $m \in \mathcal{M} \setminus \{T\}$ such that $U^{(m)}$ is not $\sigma(U^{(T)})$-measurable, and
    \[
    I_\tau(U^{(m)}; Y \mid U^{(\neg m)}) > 0 ,
    \]
    where $U^{(\neg m)} := (U^{(k)})_{k \in \mathcal{M} \setminus \{m\}}$.
    \item The cost functional is uncoupled across modalities.
\end{enumerate}
\end{assumption}

\begin{theorem}[Expectile Budget--Modality Separation]\label{thm:main}
Under Assumption~\ref{ass:conditions}, the following statements hold:

\begin{enumerate} 
    \item [(i)]  {Strict Information Gap:}
    \[
    C_\tau(B) > \sup_{P_{U^{(T)}} \in \mathcal{P}_T(B)} e_\tau(L(U^{(T)};Y)),
    \]
    where $\mathcal{P}_T(B) := \{P_{U^{(T)}} : \mathbb{E}[\mathrm{Cost}_T(U^{(T)})] \le B\}$.

    \item[(ii)]  {Modal Decomposition:} There exists an aggregation functional $\Phi_\tau: \mathbb{R}_+^{\mathcal{M}} \to \mathbb{R} $ such that
    \[
    C_\tau(B) = \sup_{\sum_{m \in \mathcal{M}} B_m \le B} \Phi_\tau\!\left( \{C_\tau^{(m)}(B_m)\}_{m \in \mathcal{M}} \right),
    \]
    where $C_\tau^{(m)}(B_m) := \sup_{\mathbb{E}[\mathrm{Cost}_m] \le B_m} e_\tau(L(U^{(m)};Y))$.

    \item[(iii)]  {Strict Incompleteness of Text-Only Representations:}
    \[
    C_\tau(B) > C_\tau^{\text{text-only}}(B)
    \]
    for all $\tau \neq \tfrac12$, where $C_\tau^{\text{text-only}}(B)$ is the capacity restricted to the sub-$\sigma$-algebra generated by $U^{(T)}$.
\end{enumerate}
\end{theorem}

\begin{proof}  We prove   Theorem~\ref{thm:main}.
We start with Part (i): Strict Information Gap.
Let $\mathcal{P}_T(B) := \{P_U \in \mathcal{P}(B) : U^{(m)} \text{ is degenerate for all } m \neq T\}$. Since $\mathcal{P}_T(B) \subset \mathcal{P}(B)$,
\[
C_\tau(B) \ge C_\tau^{\text{text-only}}(B) := \sup_{P_U \in \mathcal{P}_T(B)} I_\tau(U;Y).
\]

To establish strictness, consider the chain rule decomposition. For any $P_U \in \mathcal{P}(B)$,
\[
L(U;Y) = \log \frac{dP_{Y|U}}{dP_Y}(Y \mid U)
= L(U^{(T)};Y) + \Delta(U;Y),
\]
where $\Delta(U;Y) := \log \frac{dP_{Y|U}}{dP_{Y|U^{(T)}}}(Y \mid U)$. By Assumption A2,
\[
\mathbb{E}[\Delta(U;Y) \mid U^{(T)}] = \mathrm{KL}(P_{Y|U} \| P_{Y|U^{(T)}}) \ge 0,
\]
with strict inequality on positive measure. Thus $\Delta \ge 0$ a.s. and $\mathbb{P}(\Delta > 0) > 0$, implying $L(U;Y) \ge L(U^{(T)};Y)$ with strict inequality on positive measure. By monotonicity Lemma,
$
e_\tau(L(U;Y)) > e_\tau(L(U^{(T)};Y)).
$
Taking suprema yields $C_\tau(B) > C_\tau^{\text{text-only}}(B)$.

We now move to  {Part (ii): Modal Decomposition}

The additive cost structure permits a two-stage optimization. For any $P_U \in \mathcal{P}(B)$, define $B_m := \mathbb{E}[\mathrm{Cost}_m(U^{(m)})]$. Then $\sum_m B_m \le B$. For any budget allocation $\mathbf{B}$ with $\sum_m B_m \le B$, the admissible joint distributions factor through marginal constraints.

Let $\mathbf{C}$ denote the copula coupling the marginal distributions. The copula $\mathbf{C}: [0,1]^M \to [0,1]$ coupling the modalities is the multivariate cumulative distribution function (CDF) of the probability integral transforms of the inputs, explicitly defined by Sklar's Theorem as $\mathbf{C}(v_1, \dots, v_M) = F(F_1^{-1}(v_1), \dots, F_M^{-1}(v_M))$, where $F$ is the joint CDF of the multimodal vector $U$, and $F_m^{-1}$ represents the quantile functions of the individual marginal distributions $F_m(u^{(m)}) = \mathbb{P}(U^{(m)} \le u^{(m)})$. When differentiated, it yields the copula density $c(F_1(u^{(1)}), \dots, F_M(u^{(M)}))$, which explicitly modulates the joint density via $p(u) = \left[\prod_m p_m(u^{(m)})\right] c(F_1(u^{(1)}), \dots, F_M(u^{(M)})).$ 
Then
\[
L(U;Y) = \sum_{m \in \mathcal{M}} L(U^{(m)};Y) + R(\mathbf{C}),
\]
where $R(\mathbf{C}) := \log \frac{dP_{Y|U}}{\prod_m dP_{Y|U^{(m)}}}$. Define
\[
\Phi_\tau(\{C_\tau^{(m)}(B_m)\}_{m \in \mathcal{M}})
:= \sup_{\mathbf{C}} e_\tau\!\left( \sum_{m \in \mathcal{M}} L(U^{(m)};Y) + R(\mathbf{C}) \right).
\]
Then
\[
C_\tau(B) = \sup_{\sum_m B_m \le B} \Phi_\tau(\{C_\tau^{(m)}(B_m)\}_{m \in \mathcal{M}}).
\]

We now prove Part (iii) on  Strict Incompleteness. This follows directly from Part (i), establishing
$
C_\tau(B) > C_\tau^{\text{text-only}}(B)
$
for all $\tau \neq \tfrac12$. This demonstrates that text-only representations incur fundamental information loss under risk-sensitive criteria.
\end{proof}

Theorem~\ref{thm:main} establishes that risk-sensitivity fundamentally alters the structure of information capacity in multimodal systems. The non-additivity property is particularly striking: unlike Shannon's additive capacity decomposition, expectile-based capacity exhibits intrinsic cross-modal coupling that cannot be eliminated through optimization.

\begin{corollary}[Risk-Sensitivity Breaks Additivity]\label{cor:break}
For any $\tau \neq \tfrac12$, there exist budget allocations $\mathbf{B}^{(1)}, \mathbf{B}^{(2)}$ such that
\[
C_\tau(\mathbf{B}^{(1)} + \mathbf{B}^{(2)}) \neq C_\tau(\mathbf{B}^{(1)}) + C_\tau(\mathbf{B}^{(2)}).
\]
\end{corollary}

This has profound implications for system design, indicating that resource allocation under risk-sensitive objectives requires joint optimization across modalities rather than independent per-modality allocation.

\subsection{Dual-Expectile Budget--Modality Separation}

\begin{assumption}[Regularity Conditions]\label{ass:dual_conditions}
The following structural conditions hold for all $\tau_1, \tau_2 \in (0,1)$:
\begin{enumerate}
    \item \textbf{Moment Condition:} $\mathbb{E}_{P_{U,Y}}[|L(U;Y)|^2] < \infty$ and $\mathbb{E}_{P_U}[|\mathrm{Cost}(U)|^2] < \infty$ for all $P_U \in \mathcal{P}_{\tau_2}(B)$.
    \item \textbf{Modal Non-Redundancy:} There exists $m \in \mathcal{M} \setminus \{T\}$ such that $U^{(m)}$ is not $\sigma(U^{(T)})$-measurable, and $I_{\tau_1}(U^{(m)}; Y \mid U^{(\neg m)}) > 0$ $\mathbb{P}$-a.s.
    \item \textbf{Additive Cost Structure:} $\mathrm{Cost}(U) = \sum_{m \in \mathcal{M}} \mathrm{Cost}_m(U^{(m)})$.
\end{enumerate}
\end{assumption}

\begin{definition}
The Dual-Expectile Capacity $C_{\tau_1, \tau_2}(B)$ is the supremum of the $\tau_1$-expectile mutual information functional over the $\tau_2$-constrained space of input probability measures:
$$C_{\tau_1, \tau_2}(B) := $$ $$\sup_{P_U \in \mathcal{P}(\mathcal{U})} \left\{ I_{\tau_1}(U;Y) \ \Bigg| \ \begin{aligned} &P_{U,Y} = P_U \otimes P_{Y|U}, \ L(U;Y) \in L^2(P_{U,Y}), \\ &\inf \left\{ b \in \mathbb{R} : \int_{\mathcal{U}} \rho_{\tau_2}\big(\mathrm{Cost}(u) - b\big) \, dP_U(u) = 0 \right\} \le B \end{aligned} \right\}$$
\end{definition}

\begin{theorem}[Dual-Expectile Budget--Modality Separation]\label{thm:dual_main}
Under Assumption~\ref{ass:dual_conditions}, for any risk parameters $\tau_1, \tau_2 \in (0,1)$, the following properties hold:
\begin{enumerate}
    \item[(i)]  {Strict Information Gap:} 
    $
    C_{\tau_1, \tau_2}(B) > \sup_{P_{U^{(T)}} \in \mathcal{P}_{T, \tau_2}(B)} e_{\tau_1}(L(U^{(T)}; Y)),
    $
    where $\mathcal{P}_{T, \tau_2}(B) := \{P_{U^{(T)}} : e_{\tau_2}[\mathrm{Cost}_T(U^{(T)})] \le B\}$.
    
    \item[(ii)]  {Coupled Modal Decomposition:} There exists an aggregation functional $\Phi_{\tau_1, \tau_2}$ mapping marginal capacities and dependencies to $\mathbb{R}$ such that:
    \[
    C_{\tau_1, \tau_2}(B) = \sup_{\mathbf{B} \in \mathcal{A}_{\tau_2}(B)} \Phi_{\tau_1, \tau_2}\left( \{C_{\tau_1, \tau_2}^{(m)}(B_m)\}_{m \in \mathcal{M}} \right),
    \]
    where $C_{\tau_1, \tau_2}^{(m)}(B_m) := \sup_{e_{\tau_2}[\mathrm{Cost}_m] \le B_m} e_{\tau_1}(L(U^{(m)};Y))$, and $\mathcal{A}_{\tau_2}(B)$ is the set of effective budget vectors allocated under the $\tau_2$-expectile constraint.
    
    \item[(iii)]  {Dual Non-Additivity:} If $\tau_1 \neq \tfrac{1}{2}$ or $\tau_2 \neq \tfrac{1}{2}$, capacity scaling is non-additive for non-comonotonic variables due to either asymmetric information aggregation ($\tau_1$) or sub/super-additivity of the budget constraint ($\tau_2$).
    
   \end{enumerate}
\end{theorem}

\begin{proof}
We start with {Part (i): Proof of Strict Information Gap}.
Let $\mathcal{P}_{T, \tau_2}(B)$ be the embedded subset of $\mathcal{P}_{\tau_2}(B)$ where $U^{(m)}$ is degenerate (constant) for all $m \neq T$. By definition of the supremum over a restricted domain, we have $C_{\tau_1, \tau_2}(B) \ge \sup_{P_{U^{(T)}} \in \mathcal{P}_{T, \tau_2}(B)} e_{\tau_1}(L(U^{(T)}; Y))$. To establish strict inequality, we express the full joint information density via the conditional decomposition:
$
L(U;Y) = \log \frac{dP_{Y|U}}{dP_Y}(Y \mid U) = L(U^{(T)};Y) + \Delta(U;Y),
$
where $\Delta(U;Y) := \log \frac{dP_{Y|U}}{dP_{Y|U^{(T)}}}(Y \mid U)$. Taking the conditional expectation with respect to the textual sub-$\sigma$-algebra $\sigma(U^{(T)})$, we obtain:
$
\mathbb{E}[\Delta(U;Y) \mid U^{(T)}] = \mbox{KL}(P_{Y|U} \,\|\, P_{Y|U^{(T)}}).
$
By the Modal Non-Redundancy condition (Assumption A2), there exists an $m \neq T$ such that $U^{(m)}$ adds distinct predictive value, implying $\mbox{KL}(P_{Y|U} \,\|\, P_{Y|U^{(T)}}) \ge 0$ with strict inequality holding on a set of positive $\mathbb{P}$-measure. $\Delta(U;Y) \ge 0$ a.s. and $\mathbb{P}(\Delta(U;Y) > 0) > 0$. This guarantees that $L(U;Y) \ge L(U^{(T)};Y)$ with strict inequality on a non-null set. Because the expectile operator $e_{\tau_1}(\cdot)$ satisfies strict monotonicity for random variables ordered by strict stochastic dominance on positive measure, it follows that:
$
e_{\tau_1}(L(U;Y)) > e_{\tau_1}(L(U^{(T)};Y)).
$
Taking the supremum over $\mathcal{P}_{\tau_2}(B)$ preserves the strict inequality due to the open support guaranteed by non-redundancy, yielding the assertion.

We now prove {Part (ii) and (iii): Proof of Coupled Modal Decomposition and Non-Additivity}
Under Assumption A3, the cost is additive: $\mathrm{Cost}(U) = \sum_{m} \mathrm{Cost}_m(U^{(m)})$. However, because $\tau_2 \neq \tfrac{1}{2}$, the expectile constraint is itself non-additive. Let $B_m = e_{\tau_2}[\mathrm{Cost}_m(U^{(m)})]$. The joint budget constraint satisfies the implicit score equation:
\[
\mathbb{E}\left[ \rho_{\tau_2}\left( \sum_{m \in \mathcal{M}} \mathrm{Cost}_m(U^{(m)}) - B \right) \right] = 0.
\]
Because expectiles are sub-additive for $\tau_2 > \tfrac{1}{2}$ and super-additive for $\tau_2 < \tfrac{1}{2}$ for non-comonotonic costs, the admissible budget allocation space $\mathcal{A}_{\tau_2}(B)$ is defined by the coupling:
\[
\mathcal{A}_{\tau_2}(B) = \left\{ \mathbf{B} \in \mathbb{R}_+^M : e_{\tau_2}\left(\sum_{m \in \mathcal{M}} \mathrm{Cost}_m(U^{(m)})\right) \le B \right\} \neq \left\{ \mathbf{B} : \sum_{m \in \mathcal{M}} B_m \le B \right\}.
\]
Simultaneously, let $\mathbf{C}$ be the unique copula parsing the joint distribution $P_U$ into its marginals via Sklar's Theorem. The information density decomposes as:
\[
L(U;Y) = \sum_{m \in \mathcal{M}} L(U^{(m)};Y) + \log c(F_1(U^{(1)}), \dots, F_M(U^{(M)})),
\]
where $c$ is the copula density. We define the dual-expectile aggregation functional $\Phi_{\tau_1, \tau_2}$ as:
\[
\Phi_{\tau_1, \tau_2}\left(\{C_{\tau_1, \tau_2}^{(m)}(B_m)\}\right) := \sup_{\mathbf{C}} e_{\tau_1}\left( \sum_{m \in \mathcal{M}} L(U^{(m)};Y) + \log c \right).
\]
Now we evaluate additivity. Let $\theta = e_{\tau_1}(X+Z)$ for non-comonotonic information components $X$ and $Z$ from different modalities. The first-order characterization yields:
$
\mathbb{E}[\rho_{\tau_1}(X+Z - \theta)] = 0.
$
Assuming additivity $\theta = e_{\tau_1}(X) + e_{\tau_1}(Z) := \theta_X + \theta_Z$ and expanding via the identity $\psi_{\tau}(z) = \tau z_+ - (1-\tau)z_-$ reveals:
\[
\mathbb{E}[\rho_{\tau_1}(X+Z - \theta_X - \theta_Z)] = (2\tau_1 - 1) \Omega(X,Z),
\]
where $\Omega(X,Z) = \mathbb{E}[(X-\theta_X)(\mathbf{1}_{\{X+Z \ge \theta_X+\theta_Z\}} - \mathbf{1}_{\{X \ge \theta_X\}}) + (Z-\theta_Z)(\mathbf{1}_{\{X+Z \ge \theta_X+\theta_Z\}} - \mathbf{1}_{\{Z \ge \theta_Z\}})]$. For non-comonotonic variables, $\Omega(X,Z) \neq 0$. Thus, if $\tau_1 \neq \tfrac{1}{2}$, the right-hand side cannot vanish, breaking additivity of the objective. Similarly, if $\tau_2 \neq \tfrac{1}{2}$, the constraint boundaries warp non-linearly, preventing linear decomposition.

\end{proof}

\section{Conclusion}
This paper introduced a risk‑aware information theory based on expectiles. We proved that for asymmetry parameter $\tau \neq 1/2$, the resulting information measures are fundamentally distinct from Shannon's framework. In particular, we extended these results to define the {\it risk-aware} multi-user AWGN Multiple Access Channel capacity boundaries under asymmetric profiles. We demonstrated that risk-neutral Shannon metrics are fundamentally blind to tail-risk volatility because the expectation operator collapses the distribution before evaluating its upper-conditional tail integrals. Reaching and verifying  Superintelligence requires a complete transition to adaptive, risk-aware functional manifolds.

\section*{Biography}

 \begin{IEEEbiography}[{\includegraphics[width=1in,clip,keepaspectratio]{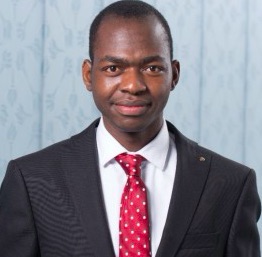}}]{Hamidou Tembine} (SM'13) is the co-founder of Timadie, Grabal  AI Mali,  co-chair of TF, founder of Guinaga, WETE, MFTG, LnG Lab, CI4SI and a Professor of Artificial Intelligence at UQTR, Quebec, Canada. He graduated in Applied Mathematics from Ecole Polytechnique (Palaiseau, France) and received the Ph.D. degree from INRIA and University of Avignon, France. He further received his Master degree in game theory and economics. His main research interests are learning, evolution, and games. In 2014, Tembine received the IEEE ComSoc Outstanding Young Researcher Award for his promising research activities for the benefit of society. He was the recipient of 10+ best paper awards in the applications of game theory. Tembine is a prolific researcher and holds 300+ scientific publications including magazines, letters, journals and conferences. He is author of the book on ``distributed strategic learning for engineers" (published at CRC Press, Taylor \& Francis 2012) which received book award 2014, and co-author of the book ``Game Theory and Learning in Wireless Networks'' (Elsevier Academic Press) and co-author of the book on ``Mean-Field-Type Games I-II " (Springer Nature)  and  for Engineers' (CRC Press0, and the author of the book ``GPT Meets Game Theory".  Tembine has been co-organizer of several scientific meetings on game theory in agriculture, water, food, environment, networking, wireless communications and smart energy systems. He has been a visiting researcher at University of California at Berkeley (US), University of McGill (Montreal, Quebec, Canada), University of Illinois at Urbana-Champaign (UIUC, US), Ecole Polytechnique Federale de Lausanne (EPFL, Switzerland) and University of Wisconsin (Madison, US). He has been a Simons Participant and a Senior Fellow 2020. He is a senior member of IEEE. He is a Next Einstein Fellow, Class of 2017.
 \end{IEEEbiography}


\begin{thebibliography}{10}
\bibitem{artzner1999coherent}
Artzner, Philippe, Freddy Delbaen, Jean‐Marc Eber, and David Heath. "Coherent measures of risk." Mathematical finance 9, no. 3 (1999): 203-228.

\bibitem{newey1987asymmetric}
Newey, Whitney K., and James L. Powell. "Asymmetric least squares estimation and testing." Econometrica: Journal of the Econometric Society (1987): 819-847.

\bibitem{renyi1961measures}
Rényi, Alfréd. "On measures of entropy and information." In Proceedings of the fourth Berkeley symposium on mathematical statistics and probability, volume 1: contributions to the theory of statistics, vol. 4, pp. 547-562. University of California Press, 1961.

\bibitem{liese2006divergences}
Liese, Friedrich, and Igor Vajda. "On divergences and informations in statistics and information theory." IEEE Transactions on Information Theory 52, no. 10 (2006): 4394-4412.

\bibitem{hung2020expectation}
Hu, Zhaolin, and L. Jeff Hong. "Kullback-Leibler divergence constrained distributionally robust optimization." Available at Optimization Online 1, no. 2 (2013): 9.

\bibitem{bellini2014generalized}
Bellini, Fabio, Bernhard Klar, Alfred Müller, and Emanuela Rosazza Gianin. "Generalized quantiles as risk measures." Insurance: Mathematics and Economics 54 (2014): 41-48.





\bibitem{basar2026a}
T. Başar, B. Djehiche, and H. Tembine.
\newblock {\em Mean-Field-Type Game Theory I: Foundations and New Directions}.
\newblock Birkhäuser, Cham, 2026.

\bibitem{basar2026b}
T. Başar, B. Djehiche, and H. Tembine.
\newblock {\em Mean-Field-Type Game Theory II: Applications}.
\newblock Birkhäuser, Cham, 2026.

\bibitem{audioiamali}
Hamidou Tembine and
                  Issa Bamia and
                  Massa Ndong and
                  Bakary Coulibaly and
                  Oumar Issiaka Traore and
                  Moussa Traore and
                  Moussa Sanogo and
                  Mamadou Eric Sangare and
                  Salif Kante and
                  Daryl Noupa Yongueng and
                  Hafiz Tiomoko Ali and
                  Malik Tiomoko and
                  Frejus Laleye and
                  Boualem Djehiche and
                  Wesmanegda Elisee Dipama and
                  Idris Baba Saje and
                  Hammid Mohammed Ibrahim and
                  Moumini Sanogo and
                  Marie Coursel Nininahazwe and
                  Abdul{-}Latif Siita and
                  Haine Mhlongo and
                  Teddy Nelvy Dieu Merci Kouka and
                  Mariam Serine Jeridi and
                  Mutiyamuogo Parfait Mupenge and
                  Lekoueiry Dehah and
                  Abdoul{-}Aziz Bio Sidi D. Bouko and
                  Wilfried Franceslas Zokoue and
                  Odette Richette Sambila and
                  Alina RS Mbango and
                  Mady Diagouraga and
                  Oumarou Moussa Sanoussi and
                  Gizachew Dessalegn and
                  Mohamed Lamine Samoura and
                  Bintou Laetitia Audrey Coulibaly,
{Breaking the Barriers of Text-Hungry and Audio-Deficient AI},
{2025},
{arXiv},
{2506.02443},

\bibitem{audioiamali2}
Hamidou Tembine: The Ghost in the Index: Knowledge Exclusion and the Fallacy of the Low-Resource Label: A Position paper,  Submitted manuscript, 2026.

\end{thebibliography}
\end{document}